\documentclass[journal]{IEEEtran}

\usepackage[T1]{fontenc}

\usepackage{cite}
\usepackage{afterpage}

\ifCLASSINFOpdf
  \usepackage[pdftex]{graphicx}
  \usepackage{epstopdf}
\else
  \usepackage[dvips]{graphicx}
\fi

\usepackage[cmex10]{amsmath}
\usepackage{bbm}
\usepackage{amssymb}
\usepackage{mathrsfs}
\allowdisplaybreaks
\usepackage{amsthm}
\usepackage[cmintegrals]{newtxmath}

\usepackage{algorithmic}
\usepackage[algoruled, linesnumbered]{algorithm2e} 

\usepackage{array}

\ifCLASSOPTIONcompsoc
  \usepackage[caption=false,font=normalsize,labelfont=sf,textfont=sf]{subfig}
\else
\usepackage[caption=false,font=footnotesize]{subfig}
\fi

\usepackage{fixltx2e}

\usepackage{stfloats}

\usepackage{url}

\usepackage{xpatch}
\usepackage{tikz}
\usetikzlibrary{patterns}
\usepackage{setspace}
\usepackage{amsthm}
\allowdisplaybreaks
\usepackage[normalem]{ulem}
\usepackage{color}

\usepackage{xpatch}
\usepackage{tikz}
\usetikzlibrary{patterns}
\usepackage{setspace}
\usepackage{amsthm}
\allowdisplaybreaks
\usepackage[normalem]{ulem}
\usepackage{color}
\usepackage{scalerel}
\newlength\bshft
\def\fakebold#1{\ThisStyle{\ooalign{$\SavedStyle#1$\cr%
  \kern-\bshft$\SavedStyle#1$\cr%
  \kern\bshft$\SavedStyle#1$}}}

\newtheorem{lemma}{Lemma}[section]

\newtheorem{remark}{Remark}[section]
\allowdisplaybreaks

\begin{document}

\hyphenation{op-tical net-works semi-conduc-tor}

\markboth{To appear in IEEE Transactions on Wireless Communications, 2019}
{Shell \MakeLowercase{\textit{et al.}}: Bare Demo of IEEEtran.cls for Journals}

\title{Continuous Analog Channel Estimation Aided Beamforming for Massive MIMO Systems}
\author{Vishnu~V.~Ratnam,~\IEEEmembership{Member,~IEEE,}
and Andreas~F.~Molisch,~\IEEEmembership{Fellow,~IEEE,}
\thanks{V. V. Ratnam is with the Standards and Mobility Innovation Lab at Samsung Research America, Plano, TX, 75023 USA. A. F. Molisch is with the Ming Hsieh Department of Electrical and Computer Engineering, University of Southern California, Los Angeles, CA, 90089 USA (e-mail: \{ratnam, molisch\}@usc.edu). This work was supported by the National Science Foundation. Part of this work was presented at IEEE ICC 2018 \cite{Ratnam_ICC2018}.} 
}

\IEEEtitleabstractindextext{%
\begin{abstract}
Analog beamforming greatly reduces the implementation cost of massive antenna transceivers by using only one up/down-conversion chain. However, it incurs a large pilot overhead when used with conventional channel estimation (CE) techniques. This is because these CE techniques involve digital processing, requiring the up/down-conversion chain to be time-multiplexed across the antenna dimensions. This paper introduces a novel CE technique, called continuous analog channel estimation (CACE), that avoids digital processing, enables analog beamforming at the receiver and additionally provides resilience against oscillator phase-noise. By avoiding time-multiplexing of up/down-conversion chains, the CE overhead is reduced significantly and furthermore becomes independent of the number of antenna elements. In CACE, a reference tone is transmitted continuously with the data signals, and the receiver uses the received reference signal as a matched filter for combining the data signals, albeit via analog processing. We propose a receiver architecture for CACE, analyze its performance in the presence of oscillator phase-noise, and derive near-optimal system parameters and power allocation. Transmit beamforming and initial access procedure with CACE are also discussed. Simulations confirm that, in comparison to conventional CE, CACE provides phase-noise resilience and a significant reduction in the CE overhead, while suffering only a small loss in signal-to-interference-plus-noise-ratio. 
\end{abstract}

}

\maketitle

\IEEEdisplaynontitleabstractindextext

\IEEEpeerreviewmaketitle

\section{Introduction} \label{sec_intro}
Massive multiple-input-multiple-output (MIMO) systems, where the transmitter (TX) and/or receiver (RX) are equipped with a large array of antenna elements, are considered a key enabler of 5G cellular technologies due to the massive beamforming and/or spatial multiplexing gains they offer \cite{Marzetta2010, Boccardi2014}. This technology is especially attractive at millimeter (mm) wave and terahertz (THz) frequencies, where the massive antenna arrays can be built with small form factors, and where the resulting beamforming gain can help compensate for the large channel attenuation.
Despite the numerous benefits, full complexity massive MIMO transceivers, where each antenna has a dedicated up/down-conversion chain, are hard to implement in practice. This is due to the cost and power requirements of the up/down-conversion chains -- which include expensive and power hungry circuit components such as the analog-to-digital converters (ADCs) and digital-to-analog converters \cite{Murmann_ADC_compiled}. A key solution to reduce the implementation costs of massive MIMO while retaining many of its benefits is Hybrid Beamforming \cite{Molisch_HP_mag, Heath2016, Alkhateeb2015, Sohrabi2016, Ratnam_HBwS_jrnl, Park2017}, wherein a massive antenna array is connected to a smaller number of up/down-conversion chains via the use of analog hardware, such as phase-shifters and switches. While being comparatively cost and power efficient, the analog hardware can focus the transmit/receive power into the dominant channel directions, thus minimizing the performance loss in comparison to full complexity transceivers. In this paper, we focus on a special case of hybrid beamforming with only one up/down-conversion chain (for the in-phase and quadrature signal components each), referred to as analog beamforming.

A major challenge for analog beamforming (and also hybrid beamforming in general) is the acquisition of channel state information required for beamforming, henceforth referred to as rCSI. The rCSI may include, for example, average/statistical channel parameters in some beamformer designs \cite{Sudarshan2006, Caire2017, Park2017} and instantaneous parameters in some other designs \cite{Alkhateeb2015, Sohrabi2016, Vishnu_ICC2017}. The rCSI is commonly obtained by transmitting known signals (pilots) at the TX and performing channel estimation (CE) at the RX, at least once per rCSI coherence time $T_{\rm rcoh}$ -- which is the duration for which the rCSI remains approximately constant. 
Since one down-conversion chain has to be time-multiplexed across the RX antennas for CE in analog beamforming, several pilot re-transmissions are required for rCSI acquisition \cite{Alkhateeb2014, Jeong2015, Vishnu_Globecom, Caire2017}. 
As an example, exhaustive CE approaches \cite{Jeong2015} require ${\rm O}(M_{\rm tx} M_{\rm rx})$ pilots per $T_{\rm rcoh}$, where $M_{\rm tx}, M_{\rm rx}$ are the number of TX and RX antennas, respectively and ${\rm O}(\cdot)$ represents the scaling behavior in the big-`o' notation. 
Such a large pilot overhead may consume a significant portion of the time-frequency resources when $T_{\rm rcoh}$ is short, such as in vehicle-to-vehicle channels \cite{molisch2009survey}, in systems using narrow TX/RX beams, e.g., massive MIMO systems, or in channels with large carrier frequencies (high Doppler) and high blocking probabilities, e.g., at mm-wave, THz frequencies \cite{Akdeniz2014}. 
The overhead also increases system latency and makes the initial access\footnote{Initial access refers to the process wherein, a user equipment and base-station discover each other, synchronize, and coordinate to initiate communication. \label{note0}} (IA) procedure cumbersome \cite{Barati2015, Li2016, Giordani2016_iter}. As a solution, several fast CE approaches have been proposed in literature, which are discussed below assuming $M_{\rm tx}=1$ for convenience.\footnote{For $M_{\rm tx} > 1$, the pilot overhead increases further, either multiplicatively or additively, by a function of $M_{\rm tx}$, determined by the CE algorithm used at the TX. \label{note1}} 
Side information aided CE approaches utilize spatial/temporal statistics of rCSI to reduce the pilot overhead \cite{Love2014, Adhikary2014, Ratnam_HBwS_jrnl, Vishnu_Globecom}. Compressed sensing based CE approaches \cite{Chi2013, Alkhateeb2014, Park2016_asilomar, Lee2016} exploit the sparse nature of the channel to reduce the number of pilots per $T_{\rm rcoh}$ up to ${\rm O}\big(L \log(M_{\rm rx}/L)\big)$, where $L$ is the number of non-zero components of the channel in a certain basis. Iterative angular domain CE uses progressively narrower search beams at the RX to reduce the required pilots per $T_{\rm rcoh}$ to ${\rm O}(\log M_{\rm rx})$ \cite{kim2014fast, Desai2014, Giordani2016_iter}. Approaches that utilize side information to enhance iterative angular domain CE \cite{Giordani2016_context, Devoti2016} or perform angle domain tracking \cite{Gao2017, Li2017, Alexandropoulos2017} have also been considered. Sparse ruler based approaches exploit the possible Toeplitz structure of the spatial correlation matrix to reduce pilots per $T_{\rm rcoh}$ to ${\rm O}(\sqrt{M_{\rm rx}})$ \cite{Pal2010, Pal2011, Romero2013, Rial2015_icassp, Caire2017}. For $M_{\rm tx} > 1$, matrix completion based techniques \cite{Vlachos2018} use the low channel rank to reduce pilots per $T_{\rm rcoh}$ up to $ O\big(L (M_{\rm rx}+M_{\rm tx})\big)$, where $L$ is the channel rank. In all these approaches, since the required pilots per $T_{\rm rcoh}$ still scale with $M_{\rm rx}$, they are only partially successful in reducing the CE overhead. Some of these approaches also require side information and/or prior timing/frequency synchronization \cite{Nasir2016, Meng2017}, making them less suitable for the IA procedure. Some approaches also require a long $T_{\rm rcoh}$ that spans the pilot re-transmissions and/or are only applicable for certain antenna array configurations and channel models. Compressed sensing, sparse ruler, and matrix completion approaches may further incurr large computation and/or memory overheads, making them unsuitable for use at user equipments (UEs). Finally, since these CE techniques require time-multiplexing of the up/down-conversion chains, they are prone to the transient effects of the analog hardware \cite{Sands2002}. 

The main reason for the overhead is that conventional CE approaches require processing in the digital domain, while the RX has only one down-conversion chain. 
Prior to the growth of digital hardware and digital processing capabilities, some legacy systems \cite{Breese1961, Ghose1964, Thompson1976, Golliday1982} used an alternate RX beamforming approach in single path channels (for space communication), that did not require digital processing. In this approach, an analog phase locked loop (PLL) is used to recover the carrier tone from the received signal at each RX antenna, and the recovered carrier is then used for down-converting the received signal at that antenna to base-band. Since the carrier and data experience the same inter-antenna phase shift (in single path channels), the down-conversion leads to compensation of this phase shift, enabling coherent combining of the signals from each antenna (i.e., beamforming).  
Note that carrier recovery at each RX antenna can be interpreted as an implicit estimation of the channel phase using analog hardware.\footnote{The difference between `isolation/recovery' and `estimation' is somewhat vague in the context of analog signal processing. \label{note_vague}} We shall refer to such schemes that use only analog hardware to acquire the rCSI as \emph{analog channel estimation} (ACE) techniques. 
As ACE does not involve digital processing, it avoids time-multiplexing of the down-conversion chain and shows potential in reducing the CE overhead for analog beamforming. 
The delay domain counterpart of ACE was also explored for single antenna ultra-wideband systems, referred to as transmit reference schemes \cite{Goeckel2007, Ratnam_Globecom2017}. 
However, these legacy ACE systems only exploit the carrier phase but not its amplitude, and thus are not directly applicable to multi-path channels.
Additionally, they involve recovery of the carrier at the RX via a PLL, which is difficult at the low signal-to-noise ratios (SNRs) and high frequencies encountered in mm-wave/THz systems \cite{Proakis2008, Viterbi_book}. The PLL aided recovery may also lead to a high RX phase-noise \cite{Proakis2008, Viterbi_book}, viz., random fluctuation in the instantaneous frequency of the recovered carrier, that degrades the system performance. Finally, prior works do not perform a detailed performance analysis of ACE or explore optimal system parameters. 
Therefore in this paper we explore a more generalized ACE approach for RX beamforming, called continuous ACE (CACE).  
Instead of using PLLs for carrier recovery, the CACE RX uses a local oscillator and low-pass filter combination to isolate/filter-out a received reference/carrier signal. This (i) enables exploiting both the amplitude and phase information of the channel response, which is essential for multi-path channels, (ii) avoids the poor performance of PLL based recovery at low SNRs and (iii) helps compensate for TX/RX oscillator phase noise.

In CACE, a reference tone, i.e. a sinusoidal tone at a known frequency, is continuously transmitted along with the data by the TX, as illustrated in Fig.~\ref{Fig_illustrate_tx_signal}.\footnote{Since this reference need not be at the center frequency of the TX signal, we don't refer to it as the carrier here.} At the RX, the received signal at each antenna is converted to base-band by a bank of mixers and a local oscillator that is tuned (approximately) to the reference frequency, as illustrated in Fig.~\ref{Fig_block_diag_CACE}. The in-phase (I) and quadrature (Q) components of the resulting base-band signal at each antenna are then low-pass filtered to extract the received signals corresponding to the reference, as illustrated in Fig.~\ref{Fig_illustrate_rx_signal}.  
These filtered outputs, which are implicit $\text{estimates}^{3}$ of the channel response (including amplitude and phase) at the reference frequency, are then used as control signals to a variable gain, analog phase-shifter array to generate the RX analog beam. The un-filtered base-band received signals at each antenna are processed by these phase-shifters, added and fed to a single ADC for demodulation.  
As shall be shown, this process emulates using the received signal for the reference as a matched filter for the received data signals, and it achieves a large RX beamforming gain in sparse, wide-band massive MIMO channels. This is because, while the reference and data signals have different frequencies and thus may experience different channel responses, the channel response exhibits a similar spatial signature across frequency, especially for large antenna arrays (see Remark \ref{rem_array_orth}). Furthermore, since any TX/RX oscillator phase-noise affects both the reference and data similarly, the match filtering helps partially mitigate the phase-noise from the demodulation outputs. 
As digital processing and time multiplexing of the down-conversion chain are not required, the computational complexity is low and the CE overhead does not scale with $M_{\rm rx}$. 
Unlike conventional beamformer designs \cite{Sudarshan2006, venkateswaran2010analog, Alkhateeb2015, Li2017}, CACE aided beamforming also improves diversity against multi-path component (MPC) blocking by combining the received signal power from many channel MPCs. 
The phase shifts during the receive mode can also be utilized for transmit beamforming on the reverse link. By providing an option for digitally controlling these phase-shifter inputs, the proposed architecture can also support conventional RX beamforming approaches when required. 

\begin{figure*}[h]
\vspace{-0.2cm} 
\setlength{\unitlength}{0.08in} 
\centering 
\begin{tikzpicture}[x=0.08in,y=0.08in] 
\filldraw[color=black, fill=black](-3,7) (-3,7) -- (-2,8) -- (-4,8) -- cycle;
\draw (-3,7) -- (-3,5) -- (-2,5); \draw(-2,4) rectangle (2,6) node[pos=.5] {\small LNA}; \draw(2,4) rectangle (6,6) node[pos=.5] {\small BPF}; \draw[->] (6,5) -- (22,5) -- (22,4); \draw(9.8,6) node{\small $s_{{\rm rx},M_{\rm rx}}(t)$}; 
\draw[line width = 0.4 mm] (14.5,1) -- (14.5,24);
\filldraw[fill=black] (20,3) circle(0.3); \draw[line width = 0.4 mm] (20,1) -- (20,24); \draw[->] (20,3) -- (21,3); 
\draw (22,3) circle(1); \draw(22,3) node{$\times$}; \draw (23,3) -- (23.5,3); \draw(23.5,2) rectangle (27.5,4) node[pos=.5] {\small LPF}; \draw[->] (27.5,3) -- (29,3) -- (29,1) -- (30,1); \draw(30,0) rectangle (35,2) node[pos=.5] {\small ${\rm LPF}_{\hat{g}}$}; 

\filldraw[fill=black] (14.5,7) circle(0.3); \draw[->] (14.5,7) -- (21,7); \draw[->] (22,5) -- (22,6); 
\draw (22,7) circle(1); \draw(22,7) node{$\times$}; \draw (23,7) -- (23.5,7); \draw(23.5,6) rectangle (27.5,8) node[pos=.5] {\small LPF}; \draw[->] (27.5,7) -- (29,7) -- (29,9) -- (30,9); \draw(30,8) rectangle (35,10) node[pos=.5] {\small ${\rm LPF}_{\hat{g}}$}; 
\draw[->] (35,1) -- (38,1) -- (44,9); \draw[->] (35,9) -- (38,9) -- (44,1); 
\draw[fill = white, line width = 0.4 mm] (38,2) rectangle (45.5,8); \draw(39,8.5) node{\tiny a}; \draw(39,1.5) node{\tiny b};
\draw (29,7) -- (41,7); \draw (41,6.3) -- (41,7.7) -- (42.5,7) -- cycle; \draw(41.3,7) node{\tiny a}; \draw (38.5,7) -- (38.5, 5.8) -- (39, 5.8); \draw (39,5.1) -- (39,6.5) -- (40.5,5.8) -- cycle; \draw(39.5,5.8) node{\tiny -b};
\draw[->] (42.5,7) -- (44,7); \draw(44.5,7) circle(0.5); \draw(44.5,7) node{+}; \draw[->] (40.5,4.2) -- (41,4.2) -- (44.2,6.7); \draw[->] (45,7) -- (46,7) -- (49.4,18.3);  
\draw (29,3) -- (41,3); \draw (41, 2.3) -- (41,3.7) -- (42.5,3) -- cycle; \draw(41.3,3) node{\tiny a}; \draw(38.5,3) -- (38.5, 4.2) -- (39, 4.2); \draw (39,4.9) -- (39,3.5) -- (40.5,4.2) -- cycle; \draw(39.4,4.2) node{\tiny b}; \draw[->] (42.5,3) -- (44,3); \draw(44.5,3) circle(0.5); \draw(44.5,3) node{+}; \draw[->] (40.5,5.8) -- (41,5.8) -- (44.2,3.3); \draw[->] (45.5,3) -- (46,3) -- (49.4,14.3); 
\draw(9.5,12) node{\vdots}; \draw(25.5,12) node{\vdots}; 
\filldraw[color=black, fill=black](-3,19) (-3,19) -- (-2,20) -- (-4,20) -- cycle;
\draw (-3,19) -- (-3,17) -- (-2,17); \draw(-2,16) rectangle (2,18) node[pos=.5] {\small LNA}; \draw(2,16) rectangle (6,18) node[pos=.5] {\small BPF}; \draw[->] (6,17) -- (22,17) -- (22,16); \draw(9.5,18) node{\small $s_{{\rm rx},1}(t)$}; 
\draw(6,22) node{\small Oscillator};
\draw (11,22) circle(1); \draw(11,22) node{$\sim$}; \draw[->] (12,22) -- (16,22); \draw(16,21) rectangle (19,23) node[pos=.5] {$\pi \!/\! 2$}; \draw(19,22) -- (20,22); \filldraw[fill=black] (14.5,22) circle(0.3); 
\filldraw[fill=black] (20,15) circle(0.3); \draw[->] (20,15) -- (21,15); \draw (22,15) circle(1); \draw(22,15) node{$\times$}; \draw (23,15) -- (23.5,15); \draw(23.5,14) rectangle (27.5,16) node[pos=.5] {\small LPF}; \draw[->] (27.5,15) -- (29,15) -- (29,13) -- (30,13); \draw(30,12) rectangle (35,14) node[pos=.5] {\small ${\rm LPF}_{\hat{g}}$}; \draw(33.5,15.7) node{\tiny ${\rm Im}{[\tilde{\mathbf{s}}_{\rm rx,BB}(t)]}_{1}$}; \draw[->] (27.5,15) -- (38,15); \draw(39.3,12) node{\tiny ${\rm Im}{[\hat{\mathbf{s}}_{\rm rx,BB}(t)]}_{1}$}; 
\filldraw[fill=black] (14.5,19) circle(0.3); \draw[->] (14.5,19) -- (21,19); \draw[->] (22,17) -- (22,18); \draw (22,19) circle(1); \draw(22,19) node{$\times$}; \draw (23,19) -- (23.5,19); \draw(23.5,18) rectangle (27.5,20) node[pos=.5] {\small LPF}; \draw[->] (27.5,19) -- (29,19) -- (29,21) -- (30,21); \draw(30,20) rectangle (35,22) node[pos=.5] {\small ${\rm LPF}_{\hat{g}}$}; \draw(33.5,18) node{\tiny ${\rm Re}{[\tilde{\mathbf{s}}_{\rm rx,BB}(t)]}_{1}$}; \draw[->] (29,19) -- (38,19); 
\draw[->] (35,13) -- (38,13) -- (44,21); \draw[->] (35,21) -- (38,21) -- (44,13); \draw[fill = white](38,14) rectangle (45.5,20); \draw(41,18) node{\small Phase}; \draw(41.3,16) node{\small Shifter}; \draw[->] (45.5,15) -- (49,15); \draw[->] (45.5,19) -- (49,19); \draw(39.3,21.7) node{\tiny ${\rm Re}{[\hat{\mathbf{s}}_{\rm rx,BB}(t)]}_{1}$}; 
\draw(50,19) circle(1); \draw(50,19) node{$+$}; \draw[->] (51,19) -- (63,19); \draw(57,20) node{\small $\rm{Re}\{y(t)\}$};
\draw(50,15) circle(1); \draw(50,15) node{$+$}; \draw[->] (51,15) -- (63,15); \draw(57,16) node{\small $\rm{Im}\{y(t)\}$};
\draw[dashed] (13,-5) rectangle (28,25); \draw(20.5,-2) node{\small Baseband}; \draw(20.5,-4) node{\small conversion};
\draw[dashed] (28.5,-5) rectangle (46,25); \draw(38,-2) node{\small Amplitude \& Phase}; \draw(38,-4) node{ \small compensation};
\draw(65,15) node[rotate=90] {Digital processing};
\end{tikzpicture}
%
\caption{Block diagram of an RX with analog beamforming enabled via CACE.} 
\label{Fig_block_diag_CACE} 
\end{figure*}
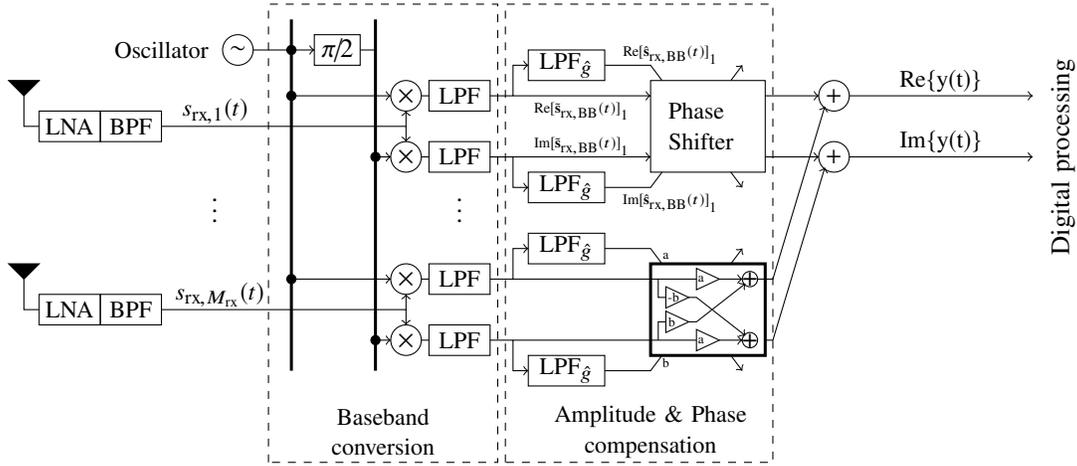 
\begin{figure}[!htb]
\centering
\vspace{-0.5cm}
\subfloat[Transmit signal at TX antenna $m$]{\includegraphics[width= 0.46\textwidth]{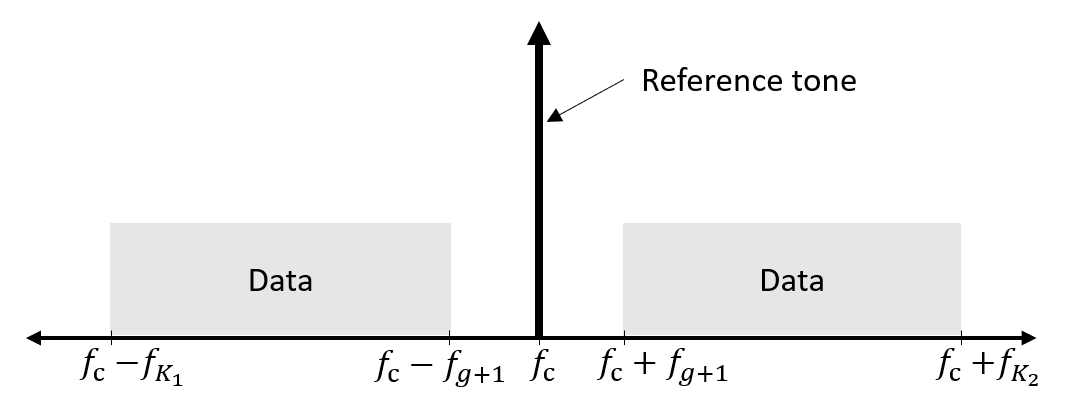} \label{Fig_illustrate_tx_signal}} \\
\subfloat[Baseband signal at RX antenna $m$]{\includegraphics[width= 0.44\textwidth]{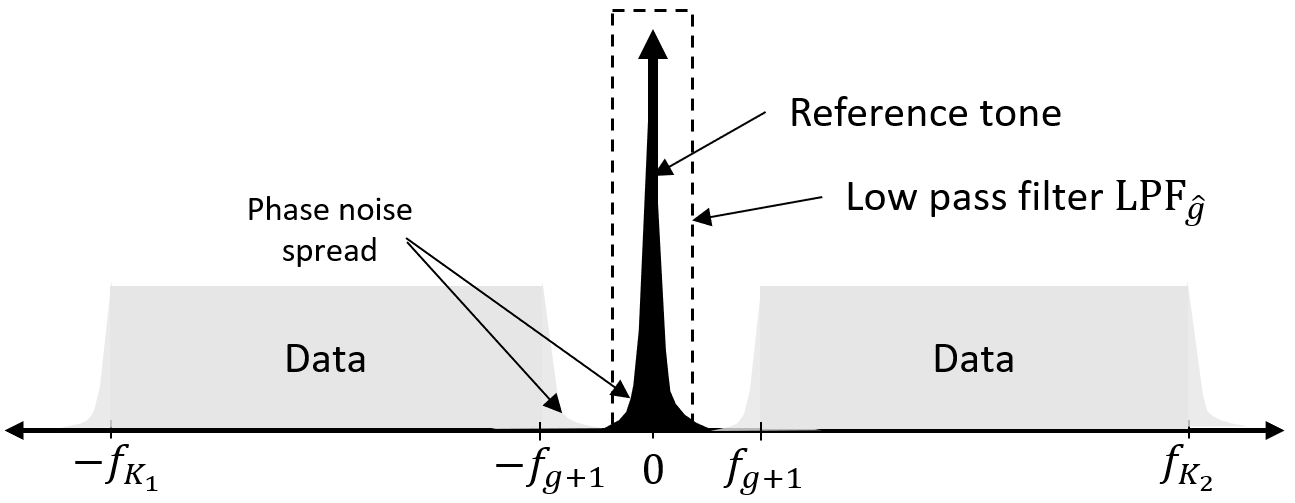} \label{Fig_illustrate_rx_signal}}
\caption{An illustration of the transmit and receive signals (in frequency domain) with CACE aided beamforming.}
\label{Fig_illustrate_signals}
\vspace{-0.8cm}
\end{figure}

On the flip side, CACE may require additional analog hardware in comparison to conventional digital CE, including $2M_{\rm rx}$ mixers and low-pass filters.  
Additionally, the accumulation of power from multiple MPCs, while improving diversity, may cause performance degradation in rich scattering channels, as shall be shown. An improperly designed reference tone may also cause strong inter-carrier interference (ICI) to the data signals due to phase noise. Finally, the proposed approach in its suggested form does not support reception of multiple spatial data streams and can only be used for beamforming at one end of a communication link. This architecture is therefore more suitable for use at the UEs. The possible extensions to multiple spatial stream reception shall be explored in future work.  

Some preliminary results of this work were published in the conference paper \cite{Ratnam_ICC2018}, albeit under a simplified RX model and without a detailed phase-noise analysis. 
As shall be shown here, the phase-noise analysis plays a vital role in CACE performance and system design. 
A different ACE technique that does not require continuous transmission of the reference, called periodic ACE (PACE), was proposed by us in \cite{Ratnam_PACE}. While PACE prevents wastage of transit resources on a continuous reference, it involves PLL aided reference recovery and suffers severe performance degradation at very low SNRs. 
A third ACE technique, called multi-antenna frequency shift reference (MA-FSR), that uses square law components instead of filters and phase-shifters at the RX, was explored by us in \cite{Ratnam_Globecom2018}. While being resilient to phase-noise like CACE and having a low hardware cost, MA-FSR is a non-coherent technique that suffers from a poor bandwidth efficiency of $50\%$. It should be emphasized that CACE, PACE and MA-FSR are three different ACE schemes to help reduce CE overhead in massive MIMO systems, each having their unique advantages and RX architectures, and each requiring separate performance analysis techniques. The detailed analysis presented in this paper for CACE, in combination with the analysis of PACE and MA-FSR in \cite{Ratnam_PACE, Ratnam_Globecom2018}, shall aid in a detailed comparison of these schemes for a specific application. 
The contributions of this paper are as follows:
\begin{enumerate}
\item We propose a novel transmission technique called CACE, that enables RX beamforming with low CE overhead, and propose a corresponding RX architecture for CACE.
\item We analytically characterize the achievable system throughput with CACE aided beamforming in a wide-band channel, and derive near-optimal system parameters. 
\item In the process, the impact of oscillator phase-noise on performance and the ability of CACE to partially suppress phase-noise are explored.
\item Simulations under practically relevant channel models are presented to support the analytical results.
\end{enumerate}
The organization of the paper is as follows: the system model is presented Section \ref{sec_chan_model}; the signal, noise and interference components at the demodulation outputs are analyzed in Section \ref{sec_demod_anal}; the system performance is characterized in Section \ref{sec_perf_anal}; IA and transmit beamforming are discussed in Section \ref{sec_IA_aCSI_at_BS}; simulations results are presented in Section \ref{sec_sim_results} and finally conclusions are drawn in Section \ref{sec_conclusions}. 

\textbf{Notation:} scalars are represented by light-case letters; vectors by bold-case letters; and sets by calligraphic letters. Additionally, ${\rm j} = \sqrt{-1}$, $a^*$ is the complex conjugate of a complex scalar $a$, $|\mathbf{a}|$ represents the $\ell_2$-norm of a vector $\mathbf{a}$, $\mathbf{A}^{\rm T}$ is the transpose of a matrix $\mathbf{A}$ and ${\mathbf{A}}^{\dag}$ is the conjugate transpose of a complex matrix $\mathbf{A}$. Finally, $\fakebold{\mathbb{I}}_a$ is an $a \times a$ identity matrix, $\fakebold{\mathbb{O}}_{a,b}$ is the $a \times b$ all zeros matrix, $\mathbb{E}\{\}$ represents the expectation operator, $\stackrel{\rm d}{=}$ represents equality in distribution, $\mathrm{Re}\{\cdot\}$/$\mathrm{Im}\{\cdot\}$ refer to the real/imaginary component, respectively, and $\mathcal{CN}(\mathbf{a},\mathbf{B})$ represents a circularly symmetric complex Gaussian vector with mean $\mathbf{a}$ and covariance matrix $\mathbf{B}$. 

\section{General Assumptions and System model} \label{sec_chan_model}
We consider a single cell system in downlink, where a $M_{\rm tx}$ antenna base-station (BS) transmits data to multiple UEs simultaneously via spatial multiplexing. 
Since we mainly focus on the downlink, we shall use the terms BS/TX and UE/RX interchangeably. Each UE is assumed to have a hybrid architecture, with $M_{\rm rx}$ antennas and one down-conversion chain, and it performs CACE aided RX beamforming. On the other hand, the BS may have an arbitrary architecture and it transmits a single spatial data-stream to each scheduled UE. For convenience, we consider the use of noise-less and perfectly linear antennas, filters, amplifiers and mixers at the BS and UEs. 
We assume the downlink BS-UE communication to be divided into three stages: (i) Initial Access (IA) (see footnote \ref{note0} on page \pageref{note0}) 
(ii) TX beamformer design - where the TX acquires rCSI for all the UEs and uses it to perform UE scheduling, TX beamforming and power allocation, and (iii) Data transmission - wherein the BS transmits data signals and the scheduled UEs use CACE to adapt the RX beams and receive the data. 
Through a major portion of this paper, we assume that the IA and TX beamformer design have been performed apriori and shall focus on the data transmission stage. 
However in Section \ref{sec_IA_aCSI_at_BS}, we shall also discuss how CACE beamforming can help in stages (i) and (ii). 

In stage (iii), we assume the BS to transmit spatially orthogonal signals to the scheduled UEs to mitigate inter-user interference. This can be achieved, for example, by careful UE scheduling and/or via avoiding transmission to common channel scatterers \cite{Adhikary2014}. For this system model and for a given TX beamformer and power allocation, we shall restrict the analysis to a single representative UE, without loss of generality. The BS is assumed to transmit orthogonal frequency division multiplexing (OFDM) symbols to the representative UE, with $K = K_1+K_2+1$ sub-carriers indexed as $\mathcal{K} = \{-K_1,...,K_2\}$. The $0$-th sub-carrier is used as the reference tone, while data is transmitted on the $K_1-g$ lower and $K_2-g$ higher sub-carriers indexed as $\mathcal{K} \setminus \mathcal{G}$, where $\mathcal{G}=\{-g,..,0,..,g\}$ defines the non-data sub-carriers and $g$ is a design parameter. 
The remaining $2g$ sub-carriers, with indices in $\mathcal{G} \setminus \{0\}$, are blanked to act as a guard band between the reference and data sub-carriers, as illustrated in Fig.\ref{Fig_illustrate_tx_signal}.\footnote{While not considered here explicitly, the results can also be extended to the case of a single-carrier system where the reference tone manifests as an un-suppressed carrier component.} 
Since the BS can afford an accurate oscillator, by ignoring its phase-noise, the \emph{complex equivalent} transmit signal for the $0$-th OFDM symbol of stage (iii) can then be expressed as:
\begin{eqnarray}
\tilde{\mathbf{s}}_{\rm tx}(t) &=& \sqrt{\frac{2}{T_{\rm s}}} \mathbf{t}\bigg[ \sqrt{E^{(\rm r)}} + \sum_{k \in \mathcal{K}\setminus \mathcal{G}} x_k e^{{\rm j} 2 \pi f_k t} \bigg] e^{{\rm j} 2 \pi f_{\rm c} t}, \label{eqn_tx_signal}
\end{eqnarray}
for $-T_{\rm cp} \leq t \leq T_{\rm s}$, where $\mathbf{t}$ is the $M_{\rm tx}\times 1$ unit-norm TX beamforming vector for this UE (designed apriori in stage (ii)), $E^{(\rm r)}$ is the energy-per-symbol allocated to the reference tone, $x_k$ is the data signal on the $k$-th OFDM sub-carrier, $f_{\rm c}$ is the reference frequency, $f_k \triangleq k/T_{\rm s}$ represents the frequency offset of the $k$-th sub-carrier and $T_{\rm s}, T_{\rm cp}$ are the symbol duration and the cyclic prefix duration, respectively. Here we define the \emph{complex equivalent} signal such that the actual (real) transmit signal is given by $\mathbf{s}_{\rm tx}(t) = \mathrm{Re}\{\tilde{\mathbf{s}}_{\rm tx}(t)\}$. 
For the data sub-carriers ($k \in \mathcal{K}\setminus \mathcal{G}$), we assume the use of independent data streams with equal power allocation, and circularly symmetric Gaussian signaling, i.e., $x_k \sim \mathcal{CN}(0, E^{(\rm d)})$. The transmit power constraint is then given by $ E^{(\rm r)} + (K-|\mathcal{G}|) E^{(\rm d)} \leq E_{\rm s}$, where $E_{\rm s}$ is the total OFDM symbol energy (excluding the cyclic prefix).

The channel to the representative UE is assumed to have $L$ MPCs with the $M_{\rm rx} \times M_{\rm tx}$ channel impulse response matrix and its Fourier transform, respectively, given as \cite{Akdeniz2014}:
\begin{subequations} \label{eqn_channel_impulse_resp}
\begin{eqnarray} 
\mathbf{H}(t) &=& \sum_{\ell=0}^{L-1} \alpha_{\ell} \mathbf{a}_{\rm rx}(\ell) {\mathbf{a}_{\rm tx}(\ell)}^{\dag} \delta(t - \tau_{\ell}), \\
\boldsymbol{\mathcal{H}}(f) &=& \sum_{\ell=0}^{L-1} \alpha_{\ell} \mathbf{a}_{\rm rx}(\ell) {\mathbf{a}_{\rm tx}(\ell)}^{\dag} e^{- {\rm j} 2 \pi (f_{\rm c}+ f) \tau_{\ell}},
\end{eqnarray}
\end{subequations}
where $\alpha_{\ell}$ is the complex amplitude and $\tau_{\ell}$ is the delay and $\mathbf{a}_{\rm tx}(\ell), \mathbf{a}_{\rm rx}(\ell)$ are the $M_{\rm tx} \times 1$ TX and $M_{\rm rx} \times 1$ RX array response vectors, respectively, of the $\ell$-th MPC. As an illustration, the $\ell$-th RX array response vector for a uniform planar array with $M_{\rm H}$ horizontal and $M_{\rm V}$ vertical elements ($M_{\rm rx} = M_{\rm H} M_{\rm V}$) is given by $\mathbf{a}_{\rm rx}(\ell) = \bar{\mathbf{a}}_{\rm rx}\big(\psi^{\rm rx}_{\rm azi}(\ell), \psi^{\rm rx}_{\rm ele}(\ell)\big)$, where 
\begin{flalign} \label{eqn_array_response_planar}
&{[\bar{\mathbf{a}}_{\rm rx}(\psi^{\rm rx}_{\rm azi}, \psi^{\rm rx}_{\rm ele})]}_{M_{\rm V}h + v} & \nonumber \\
& = \exp{\Big\{{\rm j} 2 \pi \frac{\Delta_{\rm H} h \sin(\psi^{\rm rx}_{\rm azi})\sin(\psi^{\rm rx}_{\rm ele}) + \Delta_{\rm V} (v - 1) \cos(\psi^{\rm rx}_{\rm ele})}{\lambda} \Big\}}, \!\!\!\!\!\!\! &
\end{flalign}
for $h \in \{0,..,M_{\rm H}-1\}$ and $v \in \{1,..,M_{\rm V}\}$, $\psi^{\rm rx}_{\rm azi}(\ell)$, $\psi^{\rm rx}_{\rm ele}(\ell)$ are the azimuth and elevation angles of arrival for the $\ell$-th MPC, $\Delta_{\rm H}, \Delta_{\rm V}$ are the horizontal and vertical antenna spacings and $\lambda$ is the carrier wavelength. Expressions for $\mathbf{a}_{\rm tx}(\ell)$ can be obtained similarly. Note that in \eqref{eqn_channel_impulse_resp} we implicitly ignore the frequency variation of individual MPC amplitudes $\{\alpha_{0},..,\alpha_{L-1}\}$ and the beam squinting effects \cite{Garakoui2011}, which are reasonable assumptions for moderate system bandwidths. It is emphasized that the complete channel response (including all MPCs) however still experiences frequency selective fading. 
To prevent inter-symbol interference, we let the cyclic prefix be longer than the maximum channel delay: $T_{\rm cp} > \tau_{L-1}$. We also consider a generic temporal variation model, where the time scale at which the MPC parameters $\{\alpha_{\ell}, \mathbf{a}_{\rm tx}(\ell), \mathbf{a}_{\rm rx}(\ell), \tau_{\ell}\}$ change is much larger than the symbol duration $T_{\rm s}$.
Finally, we do not assume any distribution prior or side information on $\{\alpha_{\ell}, \mathbf{a}_{\rm tx}(\ell), \mathbf{a}_{\rm rx}(\ell),\tau_{\ell}\}$. 

Each RX antenna front-end is assumed to have a low noise amplifier (LNA) followed by a band-pass filter (BPF) that leaves the desired signal un-distorted but suppresses the out-of-band noise. The filtered signal is then converted to base-band by multiplication with the I and Q components of an RX local oscillator, as depicted in the base-band conversion block of Fig.~\ref{Fig_block_diag_CACE}. This oscillator is assumed to be independently generated at the RX (i.e., without locking to the received reference). While we model the RX oscillator to suffer from phase-noise, for ease of theoretical analysis we assume the mean RX oscillator frequency to be equal to the reference frequency $f_{\rm c}$. This assumption shall be relaxed later in the simulation results in Section \ref{sec_sim_results}. Then, from \eqref{eqn_tx_signal}-\eqref{eqn_channel_impulse_resp}, the received base-band signal for the $0$-th OFDM symbol is given by:
\begin{flalign}
&\tilde{\mathbf{s}}_{\rm rx, BB}(t) & \nonumber \\
&= {\rm LPF} \Big\{ \sum_{\ell=0}^{L-1} \alpha_{\ell} \mathbf{a}_{\rm rx}(\ell) {\mathbf{a}_{\rm tx}(\ell)}^{\dag} \mathbf{s}_{\rm tx}(t - \tau_{\ell}) \sqrt{2} e^{-{\rm j}[2 \pi f_{\rm c}t + \theta(t)]} \Big\} + \tilde{\mathbf{w}}(t) & \nonumber \\
&= \frac{1}{\sqrt{T_{\rm s}}} \bigg[ \boldsymbol{\mathcal{H}}(0)\mathbf{t} \sqrt{E^{(\rm r)}} + \!\!\!\! \sum_{k \in \mathcal{K}\setminus \mathcal{G}} \!\!\!\! \boldsymbol{\mathcal{H}}(f_k)\mathbf{t} x_k e^{{\rm j} 2 \pi f_k t} \bigg] e^{- {\rm j} \theta(t)} + \tilde{\mathbf{w}}(t), \!\!\!\! & \label{eqn_defn_rx_signal}
\end{flalign}
for $0 \leq t\leq T_{\rm s}$, where the $\mathrm{Re}/\mathrm{Im}$ parts of $\tilde{\mathbf{s}}_{\rm rx, BB}(t)$ are the outputs corresponding to the I and Q components of the RX oscillator, ${\rm LPF}$ represents the low-pass filtering in the base-band conversion procedure, $\theta(t)$ is the phase-noise process of the RX oscillator and $\tilde{\mathbf{w}}(t)$ is the $M_{\rm rx}\times 1$ complex equivalent, base-band, stationary, additive, vector Gaussian noise process, with individual entries being circularly symmetric, independent and identically distributed (i.i.d.), and having a power spectral density: $\mathcal{S}_{\rm w}(f) = \mathrm{N_0}$ for $-f_{K_1} \leq f \leq f_{K_2}$. Note that while \eqref{eqn_defn_rx_signal} is obtained by assuming no TX phase-noise, the results can be generalized under some mild constraints by treating the TX phase-noise as a part of $\theta(t)$ \cite{Mathecken2011}. We model the RX phase-noise $\theta(t)$ as a Wiener process, which is representative of a free running oscillator \cite{Piazzo2002, Wu2006, Petrovic2007}. In Appendix \ref{appdixOU}, we also discuss how the results can be extended to phase-noise modeled as an Ornstein Uhlenbeck (OU) process, which is representative of an oscillator either locked to the received reference, or synthesized from a stable low frequency source \cite{Mehrotra2002, Petrovic2007}. For the Wiener model, $\theta(t)$ is a non-stationary Gaussian process which satisfies:
\begin{eqnarray} \label{eqn_PN_model}
\frac{{\rm d}\theta(t)}{{\rm d}t} = w_{\rm \theta}(t),
\end{eqnarray}
where $w_{\theta}(t)$ is a real white Gaussian process with variance $\sigma^2_{\theta}$. 
We assume the RX to have an apriori knowledge of $\sigma_{\theta}$. 
As illustrated in Fig.~\ref{Fig_block_diag_CACE}, the base-band signal ${[\tilde{\mathbf{s}}_{{\rm rx,BB}}(t)]}_{m}$ at each RX antenna $m$ is further passed through a narrow band low-pass filter: ${\rm LPF}_{\hat{g}}$ to extract the received reference signal. For convenience, ${\rm LPF}_{\hat{g}}$ is assumed to be an ideal rectangular filter with a cut-off frequency of $f_{\hat{g}}= \hat{g}/T_{\rm s}$, 
where $\hat{g} \leq g/2$. Neglecting the contribution of the data sub-carriers to the filtered outputs (which is accurate for low phase-noise i.e., $\sigma^2_{\theta} \ll g/T_{\rm s}$), these outputs can be expressed as: 
\begin{eqnarray} \label{eqn_rx_signal_LPF}
\hat{\mathbf{s}}_{\rm rx, BB}(t) &=& \frac{1}{\sqrt{T_{\rm s}}} \boldsymbol{\mathcal{H}}(0) \mathbf{t} \sqrt{E^{(\rm r)}} A(t) + \hat{\mathbf{w}}(t),
\end{eqnarray} 
where $0\leq t \leq T_{\rm s}$, we define $A(t) \triangleq \mathrm{LPF}_{\hat{g}} \{ e^{- {\rm j} \theta(t)} \}$ and $\hat{\mathbf{w}}(t)$ is the $M_{\rm rx}\times 1$ filtered Gaussian noise with power spectral density: $\mathcal{S}_{\rm w}(f) = \mathrm{N_0}$ for $-f_{\hat{g}} \leq f \leq f_{\hat{g}}$. An illustration of this filtering operation is provided in Fig.~\ref{Fig_illustrate_rx_signal}. The aim is to use $\hat{\mathbf{s}}_{\rm rx, BB}(t)$ as the combining weights for the received data signal.
This is accomplished by using $\hat{\mathbf{s}}_{\rm rx, BB}(t)$ as the control signals to a variable gain phase-shifter array, through which the base-band received signal vector $\tilde{\mathbf{s}}_{\rm rx,BB}(t)$ is processed, as shown in Fig.~\ref{Fig_block_diag_CACE}. We assume the filter cut-off frequency $f_{\hat{g}}$ to be small enough to allow the phase-shifters to respond to the slowly time varying control signals $\hat{\mathbf{s}}_{\rm rx, BB}(t)$.\footnote{A more detailed discussion about $\hat{g}$ is considered in Sections \ref{sec_perf_anal} and \ref{sec_sim_results}.} 
The phase-shifter outputs are then summed up to obtain $y(t) = \hat{\mathbf{s}}_{\rm rx, BB}(t)^{\dag} \tilde{\mathbf{s}}_{\rm rx, BB}(t)$, which is further fed to an ADC that samples at $K/T_{\rm s}$ samples/sec to obtain: 
\begin{flalign}
& y[n] \triangleq y(nT_{\rm s}/K) & \nonumber \\
&= \frac{1}{T_{\rm s}} \mathbf{t}^{\dag} {\boldsymbol{\mathcal{H}}(0)}^{\dag} \sqrt{E^{(\rm r)}} \bigg[ \boldsymbol{\mathcal{H}}(0) \mathbf{t} \sqrt{E^{(\rm r)}} & \nonumber \\
& \qquad \qquad + \sum_{k \in \mathcal{K}\setminus \mathcal{G}} \boldsymbol{\mathcal{H}}(f_k) \mathbf{t} x_k e^{{\rm j} 2 \pi k\frac{n}{K}} \bigg] {A^{*}[n] e^{-{\rm j} \theta[n]}} & \nonumber \\
& \ \ + \sqrt{\frac{1}{T_{\rm s}}} {\hat{\mathbf{w}}[n]}^{\dag} \bigg[ \boldsymbol{\mathcal{H}}(0) \mathbf{t} \sqrt{E^{(\rm r)}} + \!\!\!\!\sum_{k \in \mathcal{K}\setminus \mathcal{G}} \!\!\!\! \boldsymbol{\mathcal{H}}(f_k) \mathbf{t} x_k e^{{\rm j} 2 \pi k \frac{n}{K}} \bigg] e^{- {\rm j} \theta[n]} & \nonumber \\
& \ \ + \sqrt{\frac{1}{T_{\rm s}}} \mathbf{t}^{\dag} {\boldsymbol{\mathcal{H}}(0)}^{\dag} \sqrt{E^{(\rm r)}} \tilde{\mathbf{w}}[n] A^{*}[n] + {\hat{\mathbf{w}}[n]}^{\dag} \tilde{\mathbf{w}}[n], & \label{eqn_sampled_Y}
\end{flalign}
for $0\leq n < K$, where we define $A[n] \triangleq A\big(\frac{nT_{\rm s}}{K}\big)$, $\theta[n] \triangleq \theta\big(\frac{nT_{\rm s}}{K}\big)$, $\tilde{\mathbf{w}}[n] \triangleq \tilde{\mathbf{w}}\big(\frac{nT_{\rm s}}{K}\big)$ and $\hat{\mathbf{w}}[n] \triangleq \hat{\mathbf{w}}\big(\frac{nT_{\rm s}}{K}\big)$. Conventional OFDM demodulation is performed on $y[n]$ to demodulate the transmitted data signals, as analyzed in Section \ref{sec_demod_anal}. 

\section{Analysis of the demodulation outputs} \label{sec_demod_anal}
In this section, we study the demodulation of the sampled signal $y[n]$ for a representative $0$-th OFDM symbol. To this end, we first characterize the statistics of $e^{{\rm j}\theta[n]}$, $\tilde{\mathbf{w}}[n]$ and $A[n]$. The OFDM demodulation outputs are subsequently analyzed in Sections \ref{subsec_signal_anal}--\ref{subsec_noise_anal}. Note that the sampled, band-limited additive noise $\tilde{\mathbf{w}}[n]$ and the sampled RX phase-noise $e^{-{\rm j} \theta[n]}$ for $0 \leq n < K$ can be expressed using their normalized Discrete Fourier Transform (nDFT) expansions as:
\begin{subequations} \label{eqn_DFT_coeffs}
\begin{eqnarray}
\tilde{\mathbf{w}}[n] &=& \sum_{k \in \mathcal{K}} \mathbf{W}[k] e^{{\rm j} 2 \pi k n/K}, \\
e^{-{\rm j} \theta[n]} &=& \sum_{k \in \mathcal{K}} \Omega[k] e^{{\rm j} 2 \pi k n/K},
\end{eqnarray}
\end{subequations}
where $\mathbf{W}[k] = \frac{1}{K} \sum_{n=0}^{K-1} \tilde{\mathbf{w}}[n] e^{-{\rm j} 2 \pi k n/K}$ and $\Omega[k] = \frac{1}{K} \sum_{n=0}^{K-1} e^{-{\rm j} \theta[n]} e^{-{\rm j} 2 \pi k n/K}$ are the corresponding nDFT coefficients. Here nDFT is an unorthodox definition for Discrete Fourier Transform, where the normalization by $K$ is performed while finding $\mathbf{W}[k], \Omega[k]$ instead of in \eqref{eqn_DFT_coeffs}. These nDFT coefficients are periodic with period $K$ and satisfy the following lemma:
\begin{lemma} \label{Lemma_PN_properties}
The nDFT coefficients of $e^{-{\rm j} \theta[n]}$ for $0 \leq n < K$ satisfy:
\begin{subequations} \label{eqn_PN_lemma}
\begin{flalign}
& \sum_{k \in \mathcal{K}} \Omega[k] {\Omega[k+k_1]}^{*} = \delta^{K}_{0,k_1}, & \label{eqn_lemma_1} \\
& \Delta_{k_1,k_2} \triangleq \mathbb{E}\{\Omega[k_1] {\Omega[k_2]}^{*}\} & \nonumber \\ 
& \ \qquad \approx \frac{\delta_{k_1,k_2}^{K}}{K} \bigg[ \frac{1 - e^{-(\frac{\sigma_{\theta}^2 T_{\rm s} - {\rm j} 4 \pi k_1}{4} )}}{e^{\frac{\sigma_{\theta}^2 T_{\rm s} - {\rm j}4\pi k_1}{2K}} - 1} + \frac{1 - e^{-(\frac{\sigma_{\theta}^2 T_{\rm s} + {\rm j}4 \pi k_1}{4} )}}{1 - e^{-\frac{\sigma_{\theta}^2 T_{\rm s} + {\rm j}4 \pi k_1}{2K}}} \bigg], \! & \label{eqn_lemma_3} 
\end{flalign}
for arbitrary integers $k_1,k_2$, where $\delta^{K}_{a,b} = 1$ if $a=b \ ({\rm mod} \ K)$ or $\delta^{K}_{a,b} = 0$ otherwise. 
\end{subequations}
\end{lemma}
\begin{proof}
See Appendix \ref{appdix1}.
\end{proof}
To test the accuracy of the approximation in Lemma \ref{Lemma_PN_properties}, the Monte-Carlo simulations of $\Delta_{k,k}, \Delta_{k,k+1}$ and $\Delta_{k,k+100}$ for a typical phase-noise process ($-93$dBc/Hz at $10$MHz offset) are compared to \eqref{eqn_lemma_3} in Fig.~\ref{Fig_verify_DFT_theorem}. 
\begin{figure}[!htb]
\centering
\includegraphics[width= 0.45\textwidth]{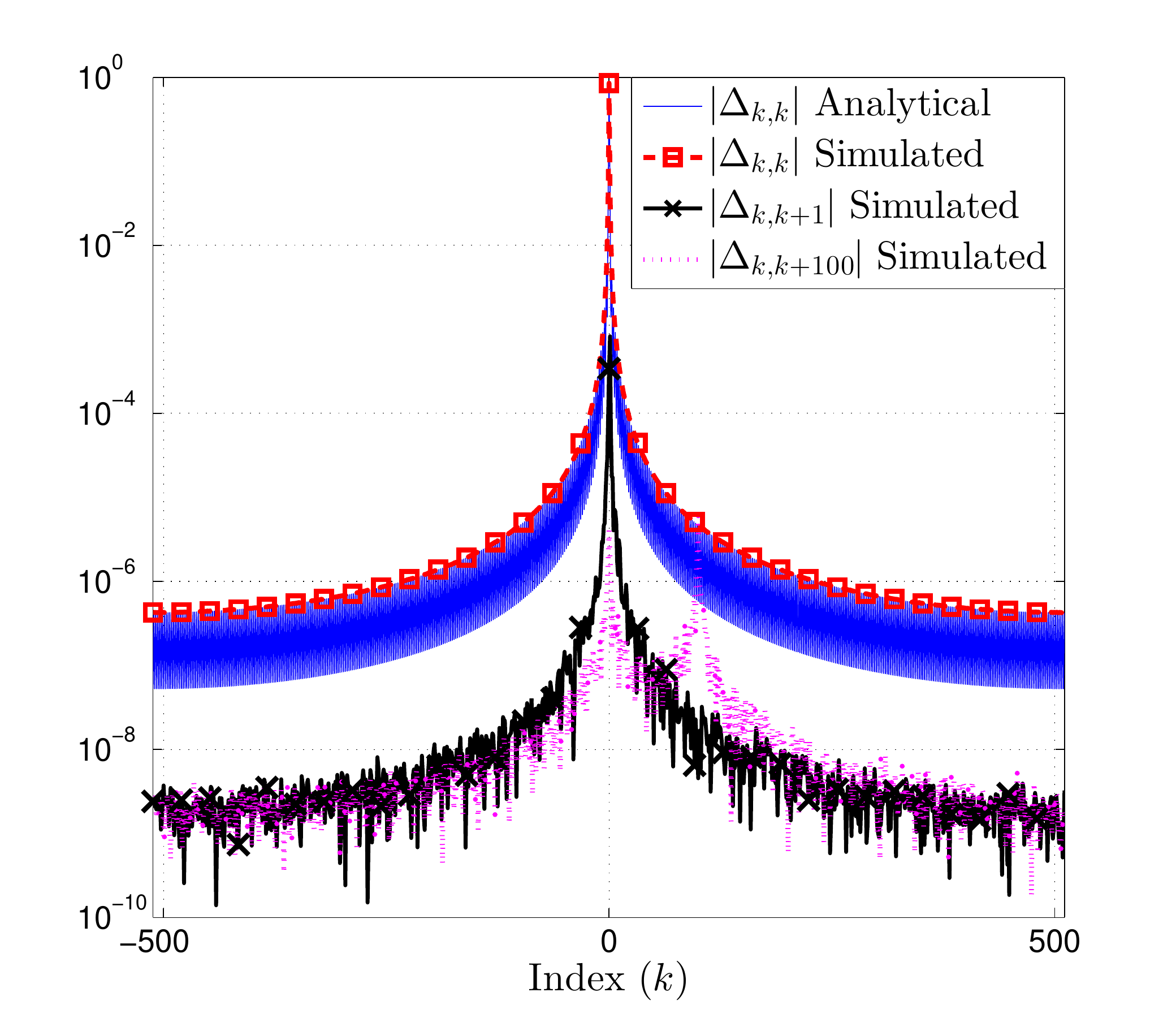}
\caption{Comparison of analytical (from Lemma \ref{Lemma_PN_properties}) and simulated statistics of the nDFT coefficients of a sample RX phase-noise process with $T_{\rm s} =1 \mu$s, $K_1=K_2+1=512$ and $\sigma_{\theta} = 1/\sqrt{T_{\rm s}}$. Simulations averaged over $10^6$ realizations.}
\label{Fig_verify_DFT_theorem}
\end{figure}
As is evident from the results, \eqref{eqn_lemma_3} is accurate for $k_1=k_2$. Similarly, the simulated values of $\Delta_{k,k+1}, \Delta_{k,k+100}$ are $\geq 20$dB lower than $\Delta_{k,k}$ $\forall k$, and can be well approximated as $0$ as in \eqref{eqn_lemma_3}. The analogous version of Lemma \ref{Lemma_PN_properties} for phase-noise modeled as an OU process is presented in Appendix \ref{appdixOU}. 
In a similar way, for the channel noise we have: 
\begin{lemma} \label{Lemma_N_properties}
The nDFT coefficients of $\tilde{\mathbf{w}}[n]$, i.e., $\big \{\mathbf{W}[k] \ \big| \ \forall k \big\}$, are jointly Gaussian with: 
\begin{subequations}
\begin{eqnarray}
\mathbb{E}\{ \mathbf{W}[k_1] {\mathbf{W}[k_2]}^{\dag}\} &=& \delta_{k_1,k_2}^{K} \frac{\mathrm{N}_0}{T_{\rm s}} \fakebold{\mathbb{I}}_{M_{\rm rx}}, \\ 
\mathbb{E}\{ \mathbf{W}[k_1] {\mathbf{W}[k_2]}^{\rm T}\} &=& \fakebold{\mathbb{O}}_{M_{\rm rx}, M_{\rm rx}}, 
\end{eqnarray}
\end{subequations}
for arbitrary integers $k_1,k_2$, where $\delta^{K}_{a,b} = 1$ if $a=b \ ({\rm mod} \ K)$ or $\delta^{K}_{a,b} = 0$ otherwise.
\end{lemma}
\begin{proof}
See Appendix \ref{appdix3}.
\end{proof}
Note that using these nDFT coefficients, the low-pass filtered versions of $\tilde{\mathbf{w}}[n]$ and $e^{-{\rm j} \theta[n]}$ in \eqref{eqn_rx_signal_LPF} can be approximated as:
\begin{subequations} \label{eqn_DFT_filt_coeffs}
\begin{eqnarray} 
\hat{\mathbf{w}}[n] &\approx & \sum_{k \in \hat{\mathcal{G}}} \mathbf{W}[k] e^{{\rm j} 2 \pi k n/K} ,\\
A[n] & \approx & \sum_{k \in \hat{\mathcal{G}}} \Omega[k] e^{{\rm j} 2 \pi k n/K} , 
\end{eqnarray}
\end{subequations}
where $\hat{\mathcal{G}} \triangleq \{-\hat{g},...,\hat{g}\}$ and the approximations are obtained by replacing the linear convolution of $\tilde{\mathbf{s}}_{\rm rx, BB}(t)$ and the filter response $\mathrm{LPF}_{\hat{g}}\{\}$ with a circular convolution. This is accurate when the filter response has a narrow support, i.e., for $\hat{g} \gg 1$. 
The remaining results in this paper are based on the approximations in \eqref{eqn_PN_lemma}--\eqref{eqn_DFT_filt_coeffs} and on an additional approximation discussed later in Remark \ref{rem_int_approx}. While we still use the $\leq, =, \geq$ operators in the following results for convenience of notation, it is emphasized that these equations are true in the strict sense only if the aforementioned approximations are met with equality. However simulation results are also used in Section \ref{sec_sim_results} to test the validity of these approximations. Substituting \eqref{eqn_DFT_coeffs} and \eqref{eqn_DFT_filt_coeffs} into \eqref{eqn_sampled_Y}, the $k$-th OFDM demodulation output can be expressed as:
\begin{flalign}
&Y_k = \frac{T_{\rm s}}{K} \sum_{n=0}^{K-1} y[n] e^{-{\rm j} 2 \pi \frac{k n}{K}} & \nonumber \\
&\quad = \sum_{\dot{k} \in \hat{\mathcal{G}}} \mathbf{t}^{\dag} {\boldsymbol{\mathcal{H}}(0)}^{\dag} \bigg( \sum_{\bar{k} \in \mathcal{K}\setminus \mathcal{G}} \boldsymbol{\mathcal{H}}(f_{\bar{k}}) \mathbf{t} \sqrt{E^{(\rm r)}} x_{\bar{k}} \Omega^{*}[\dot{k}] \Omega[\dot{k} + k-\bar{k}] & \nonumber \\
& \qquad \qquad \qquad + \boldsymbol{\mathcal{H}}(0) \mathbf{t} E^{(\rm r)} \Omega^{*}[\dot{k}] \Omega[\dot{k} + k] \bigg) & \nonumber \\
&\qquad + \sum_{\dot{k} \in \hat{\mathcal{G}}} \sqrt{T_{\rm s}} {\mathbf{W}[\dot{k}]}^{\dag} \bigg( \boldsymbol{\mathcal{H}}(0) \mathbf{t} \sqrt{E^{(\rm r)}} \Omega[k+\dot{k}] & \nonumber \\
& \qquad \qquad \qquad + \sum_{\bar{k} \in \mathcal{K}\setminus \mathcal{G}} \boldsymbol{\mathcal{H}}(f_{\bar{k}}) \mathbf{t} x_{\bar{k}} \Omega[k + \dot{k} - \bar{k}] \bigg)  & \nonumber \\
&\qquad + \sum_{\dot{k} \in \hat{\mathcal{G}}} \sqrt{T_{\rm s}} \mathbf{t}^{\dag} \boldsymbol{\mathcal{H}}(0) \mathbf{W}[k + \dot{k}] \sqrt{E^{(\rm r)}} \Omega^{*}[\dot{k}] & \nonumber \\
& \qquad + T_{\rm s} \sum_{\dot{k} \in \hat{\mathcal{G}}} \mathbf{W}^{\dag}[\dot{k}] \mathbf{W}[k + \dot{k}]. & \label{eqn_r_k} 
\end{flalign}
We shall split $Y_k$ as $Y_k = S_k + I_k + Z_k$ where $S_k$, referred to as the signal component, involves the terms in \eqref{eqn_r_k} containing $x_k$ and not containing the channel noise, $I_k$, referred to as the interference component, involves the terms containing $E^{(\rm r)}, \{x_{\bar{k}} \mid \bar{k} \in \mathcal{K} \setminus \{k\}\}$ and not containing the channel noise, and $Z_k$, referred to as the noise component, containing the remaining terms. 
These signal, interference and noise components are analyzed in the following subsections. 

\subsection{Signal Component Analysis} \label{subsec_signal_anal}
From \eqref{eqn_r_k}, the signal component for $k \in \mathcal{K}\setminus \mathcal{G}$ can be expressed as:
\begin{eqnarray} \label{eqn_S_k}
S_k &=& M_{\rm rx} \beta_{0,k} \sqrt{E^{(\rm r)}} x_{k} \Big[ \sum_{\dot{k} \in \hat{\mathcal{G}}} {|\Omega[\dot{k}]|}^2 \Big],
\end{eqnarray}
where we define $\beta_{k_1,k_2} \triangleq \mathbf{t}^{\dag} {\boldsymbol{\mathcal{H}}(k_1)}^{\dag} \boldsymbol{\mathcal{H}}(k_2) \mathbf{t} \big/ M_{\rm rx}$. 
Since $\hat{\mathcal{G}} \subset \mathcal{K}$, note that $\sum_{\dot{k} \in \hat{\mathcal{G}}} {|\Omega[\dot{k}]|}^2 < 1$ from Lemma \ref{Lemma_PN_properties}, i.e., the phase-noise causes some loss in power of the signal component. However this loss is much smaller than in PACE \cite{Ratnam_PACE} or digital CE based beamforming, where only ${|\Omega[0]|}$ contributes to the signal component unless additional phase-noise compensation is used. 
As is evident from \eqref{eqn_S_k}, CACE based beamforming utilizes the (filtered) received signal vector corresponding to the reference tone as weights to combine the received signal vector corresponding to the data sub-carriers, i.e., it emulates imperfect maximal ratio combining. 
The imperfection is because the reference tone and the $k$-th data stream pass through slightly different channels owing to the difference in their modulating frequencies. 
However the resulting loss in beamforming gain is small for sparse channels, i.e., for $L \ll M_{\rm tx}, M_{\rm rx}$, a criterion usually satisfied for mm-wave massive MIMO channels. This is due to a common channel spatial signature across frequency, as shall be shown in Sections \ref{sec_perf_anal} and \ref{sec_sim_results}. The second moment of the signal component, averaged over the phase-noise and data symbols, can be computed as:
\begin{eqnarray}
\mathbb{E}\{{|S_k|}^2\} &=& M_{\rm rx}^2 {|\beta_{0,k}|}^2 E^{(\rm r)} E^{(\rm d)} \mathbb{E} \bigg\{ {\Big[ \sum_{\dot{k} \in \hat{\mathcal{G}}}{|\Omega[\dot{k}]|}^{2} \Big]}^2 \bigg\} \nonumber \\
& \geq &  M_{\rm rx}^2 {|\beta_{0,k}|}^2 E^{(\rm r)} E^{(\rm d)} {\mu(0,\hat{g})}^2, \label{eqn_S_pow}
\end{eqnarray}
where we define $\mu(a,\hat{g}) \triangleq \sum_{\dot{k} \in \hat{\mathcal{G}}} \Delta_{a+\dot{k},a+\dot{k}}$, and \eqref{eqn_S_pow} follows from Jensen's inequality and \eqref{eqn_lemma_3}.

\subsection{Interference Component Analysis} \label{subsec_interf_anal}
From \eqref{eqn_r_k}, the interference component for $k \in \mathcal{K} \setminus \mathcal{G}$ can be expressed as:
\begin{eqnarray}
I_k &=&  \sum_{\bar{k} \in \mathcal{K} \setminus [\mathcal{G} \cup \{k\} ]} \sum_{\dot{k} \in \hat{\mathcal{G}}} M_{\rm rx} \beta_{0,\bar{k}} \sqrt{E^{(\rm r)}} x_{\bar{k}} {\Omega[\dot{k}]}^{*} \Omega[\dot{k}+k-\bar{k}] \nonumber \\
&& + \sum_{\dot{k} \in \hat{G}} M_{\rm rx} \beta_{0,0} E^{(\rm r)} {\Omega[\dot{k}]}^{*} \Omega[\dot{k}+k]. \label{eqn_I_k}
\end{eqnarray}
As is clear from above, the demodulation output $Y_k$ suffers ICI from other sub-carrier data streams and reference tone due to the RX phase-noise. 
The first and second moments of $I_k$, averaged over the other sub-carrier data $\{\bar{k} \in \mathcal{K}\setminus[\mathcal{G}\cup\{k\}]\}$ and the phase-noise, can be expressed as:
\begin{subequations}
\begin{flalign}
&\mathbb{E}\{I_k\} & \nonumber \\
& \stackrel{(1)}{=} \sum_{\dot{k} \in \hat{G}} M_{\rm rx} \beta_{0,0} E^{(\rm r)} \Delta_{\dot{k}+k,\dot{k}} = 0 & \label{eqn_I_mean} \\
& \mathbb{E}\{{|I_k|}^2\} & \nonumber \\ 
&\stackrel{(2)}{=} \!\!\sum_{\bar{k} \in \mathcal{K} \setminus [\mathcal{G} \cup \{k\} ]} \!\! M_{\rm rx}^2 {|\beta_{0,\bar{k}}|}^2 E^{(\rm r)} E^{(\rm d)} \mathbb{E} \bigg\{ {\bigg| \sum_{\dot{k} \in \hat{\mathcal{G}}} {\Omega[\dot{k}]}^{*} \Omega[\dot{k}+k-\bar{k}] \bigg|}^2 \bigg\} & \nonumber \\
& \quad + M_{\rm rx}^2 {|\beta_{0,0}|}^2 {[E^{(\rm r)}]}^2 \mathbb{E} \bigg\{ {\bigg| \sum_{\dot{k} \in \hat{\mathcal{G}}} {\Omega[\dot{k}]}^{*} \Omega[\dot{k}+k] \bigg|}^2 \bigg\} & \nonumber \\
& \stackrel{(3)}{\leq} \sum_{\bar{k} \in \mathcal{K}\setminus \{k\}} M_{\rm rx}^2 {|\beta_{\rm max}|}^2 E^{(\rm r)} E^{(\rm d)} \mathbb{E} \bigg\{ {\bigg| \sum_{\dot{k} \in \hat{\mathcal{G}}} {\Omega[\dot{k}]}^{*} \Omega[\dot{k}+k-\bar{k}] \bigg|}^2 \bigg\} & \nonumber \\
& \quad + M_{\rm rx}^2 {|\beta_{\rm max}|}^2 {[E^{(\rm r)}]}^2 \mathbb{E} \bigg\{ \Big[ \sum_{\dot{k} \in \hat{\mathcal{G}}} {\big|\Omega[\dot{k}]\big|}^2 \Big] \Big[ \sum_{\dot{k} \in \hat{\mathcal{G}}} {\big|\Omega[\dot{k}+k]\big|}^2 \Big] \bigg\} & \nonumber \\
&\stackrel{(4)}{\leq} M_{\rm rx}^2 {|\beta_{\rm max}|}^2 E^{(\rm r)} E^{(\rm d)} \mathbb{E} \bigg\{ \sum_{\dot{k},\ddot{k} \in \hat{\mathcal{G}}} {\Omega[\dot{k}]}^{*} & \nonumber \\
& \quad \times \bigg[ \sum_{\bar{k} \in \mathcal{K} \setminus \{k\}} \Omega[\dot{k}+k-\bar{k}] {\Omega[\ddot{k}+k-\bar{k}]}^{*} \bigg] \Omega[\ddot{k}] \bigg\} & \nonumber \\
& \quad + M_{\rm rx}^2 \sum_{\dot{k} \in \hat{\mathcal{G}}}  {|\beta_{\rm max}|}^2 {[E^{(\rm r)}]}^2 \mathbb{E} \big\{ {\big| \Omega[\dot{k}+k] \big|}^2 \big\} & \nonumber \\
&\stackrel{(5)}{=} M_{\rm rx}^2 {|\beta_{\rm max}|}^2 E^{(\rm r)} \! \Bigg[ E^{(\rm d)} \mathbb{E} \bigg\{ \!\sum_{\dot{k},\ddot{k} \in \hat{\mathcal{G}}} {\Omega[\dot{k}]}^{*}{\Omega[\ddot{k}]} \!-\! {\Big[ \sum_{\dot{k} \in \hat{\mathcal{G}}} {|\Omega[\dot{k}]|}^{2} \Big]}^2 \bigg\} & \nonumber \\
& \quad + \sum_{\dot{k} \in \hat{\mathcal{G}}} \! E^{(\rm r)} \mathbb{E}{\big| \Omega[\dot{k}+k] \big|}^2 \Bigg] & \nonumber \\
& \stackrel{(6)}{\leq} M_{\rm rx}^2 {|\beta_{\rm max}|}^2 E^{(\rm r)} \Big[ E^{(\rm d)} \mu(0,\hat{g}) [1 - \mu(0,\hat{g})] + E^{(\rm r)} \mu(k,\hat{g}) \Big], & \label{eqn_I_pow_1} 
\end{flalign}
\end{subequations}
where $\stackrel{(1)}{=}, \stackrel{(2)}{=}$ are obtained using the fact that $\{x_k | k \in \mathcal{K}\}$ have a zero-mean and are independently distributed; $\stackrel{(3)}{\leq}$ is obtained by defining $\beta_{\rm max} \triangleq \max_{k \in \mathcal{K}} |\beta_{k,k}|$, observing $|\beta_{0,k}| \leq \beta_{\rm max}$, and using $\mathcal{K}\setminus [\mathcal{G} \cup \{k\}] \subseteq \mathcal{K}\setminus \{k\}$ in first term and using the Cauchy-Schwarz inequality for the second term; $\stackrel{(4)}{\leq}$ follows by changing the summation order in the first term and by using \eqref{eqn_lemma_1} for the second term; $\stackrel{(5)}{=}$ follows by using $\Omega[k] = \Omega[k+K]$ and \eqref{eqn_lemma_1} for the first term and $\stackrel{(6)}{\leq}$ follows by using \eqref{eqn_lemma_3} and the Jensen's inequality. 
As shall be shown in Section \ref{sec_sim_results}, \eqref{eqn_I_pow_1} may be a loose bound on ICI for lower subcarriers, i.e., $|k| \ll K$. 
\begin{remark} \label{rem_int_approx}
A tighter approximation for $\mathbb{E}\{{|I_k|}^2\}$ can be obtained by replacing $\mu(k,\hat{g})$ in \eqref{eqn_I_pow_1} with $\tilde{\mu}(k,\hat{g}) \triangleq \sum_{\dot{k} \in \hat{\mathcal{G}}} \Delta_{\dot{k},\dot{k}} \Delta_{\dot{k}+k,\dot{k}+k}$.
\end{remark}
This heurtistic is obtained by assuming $\Omega[\dot{k}]$ and $\Omega[\ddot{k}+k]$ to be independently distributed for $\dot{k},\ddot{k} \in \hat{\mathcal{G}}$ and $k \in \mathcal{K} \setminus \mathcal{G}$ in step $\stackrel{(2)}{=}$ of \eqref{eqn_I_pow_1}, but we skip the proof for brevity. As shall be verified in Section \ref{sec_sim_results}, Remark \ref{rem_int_approx} offers a much better ICI approximation $\forall k$ and hence we shall use $\tilde{\mu}(k,\hat{g})$ instead of $\mu(k,\hat{g})$ in the forthcoming derivations in Section \ref{sec_sim_results}. 

\subsection{Noise Component Analysis} \label{subsec_noise_anal}
From \eqref{eqn_r_k}, the noise component of $Y_k$ for $k \in \mathcal{K} \setminus \mathcal{G}$ can be expressed as:
\begin{eqnarray}
Z_k &=& Z^{(1)}_k + Z^{(2)}_k + Z^{(3)}_k , \text{where:} \label{eqn_Z_k} \\
Z^{(1)}_k &=& \sum_{\dot{k} \in \hat{\mathcal{G}}} \sqrt{T_{\rm s}} {\mathbf{W}[\dot{k}]}^{\dag} \bigg( \boldsymbol{\mathcal{H}}(0) \mathbf{t} \sqrt{E^{(\rm r)}} \Omega[k+\dot{k}] \nonumber \\
&& + \sum_{\bar{k} \in \mathcal{K}\setminus \mathcal{G}} \boldsymbol{\mathcal{H}}(f_{\bar{k}}) \mathbf{t} x_{\bar{k}} \Omega[k + \dot{k} - \bar{k}] \bigg) \nonumber \\
Z^{(2)}_k &=& \sum_{\ell=0}^{L-1} \sum_{\dot{k} \in \hat{\mathcal{G}}} \sqrt{T_{\rm s}} \mathbf{t}^{\dag} {\boldsymbol{\mathcal{H}}(0)}^{\dag} \mathbf{W}[k+\dot{k}] \sqrt{E^{(\rm r)}} \Omega^{*}[\dot{k}] \nonumber \\
Z^{(3)}_k &=& T_{\rm s} \sum_{\dot{k} \in \hat{\mathcal{G}}} \mathbf{W}^{\dag}[\dot{k}] \mathbf{W}[k + \dot{k}]. \nonumber
\end{eqnarray}
Note that the noise consists of both signal-noise and noise-noise cross product terms. From Lemma \ref{Lemma_N_properties}, it can readily be verified that $\mathbb{E}\{Z_k\} = 0$ and $\mathbb{E}\{Z^{(i)}_k {[Z^{(j)}_k]}^{*}\} = 0$ for $i \neq j$, where the expectation is taken over the noise realizations. Thus the second moment of $Z_k$, averaged over the TX data, phase-noise and channel noise, can be expressed as $\mathbb{E}\{{|Z_k|}^2\} = \mathbb{E}\{{|Z_k^{(1)}|}^2\} + \mathbb{E}\{{|Z_k^{(2)}|}^2\} + \mathbb{E}\{{|Z_k^{(3)}|}^2\}$, where:
\begin{subequations} \label{eqn_noise_var_terms}
\begin{eqnarray}
\mathbb{E}\{{|Z_k^{(1)}|}^2\} & \stackrel{(1)}{=} & \sum_{\dot{k} \in \hat{\mathcal{G}}} \mathrm{N}_0 \mathbb{E}\bigg| \boldsymbol{\mathcal{H}}(0) \mathbf{t} \sqrt{E^{(\rm r)}} \Omega[\dot{k}+k] \nonumber \\
&& + \sum_{\bar{k} \in \mathcal{K} \setminus \mathcal{G}} \boldsymbol{\mathcal{H}}(\bar{k}) \mathbf{t} x_{\bar{k}} \Omega[k+\dot{k}-\bar{k}] {\bigg|}^2 \nonumber \\
& \stackrel{(2)}{=} & \sum_{\dot{k} \in \hat{\mathcal{G}}} M_{\rm rx} \mathrm{N}_0 \bigg[ \beta_{0,0} E^{(\rm r)} \Delta_{\dot{k}+k, \dot{k}+k} \nonumber \\
&& + \sum_{\bar{k} \in \mathcal{K} \setminus \mathcal{G}} \beta_{\bar{k}, \bar{k}} E^{(\rm d)}_{\bar{k}} \Delta_{k+\dot{k}-\bar{k},k+\dot{k}-\bar{k}} \bigg] \\
\mathbb{E}\{ {|Z_k^{(2)}|}^2\} & \stackrel{(3)}{=} & \sum_{\dot{k} \in \hat{\mathcal{G}}} \mathrm{N}_0 {|\boldsymbol{\mathcal{H}}(0) \mathbf{t} |}^2 E^{(\rm r)} \mathbb{E}\{ {|\Omega[\dot{k}]|}^2 \} \nonumber \\
&\stackrel{(4)}{=} & \sum_{\dot{k} \in \hat{\mathcal{G}}} M_{\rm rx} \beta_{0,0} \mathrm{N}_0 E^{(\rm r)} \Delta_{\dot{k}, \dot{k}} \\
\mathbb{E}\{{| Z_k^{(3)}|}^2\} &=& \sum_{\dot{k}, \ddot{k} \in \hat{\mathcal{G}}} T_{\rm s}^2 \mathbb{E} \Big\{ {\mathbf{W}[\dot{k}]}^{\dag} \mathbf{W}[k+\dot{k}] {\mathbf{W}[k+\ddot{k}]}^{\dag} \mathbf{W}[\ddot{k}] \Big\} \nonumber \\
& \stackrel{(5)}{=} & M_{\rm rx} |\hat{\mathcal{G}}| \mathrm{N}_0^2,
\end{eqnarray}
\end{subequations}
where $|\hat{\mathcal{G}}|=2\hat{g}+1$, $\stackrel{(1)}{=}, \stackrel{(3)}{=}$ follow from Lemma \ref{Lemma_N_properties}; $\stackrel{(2)}{=}, \stackrel{(4)}{=}$ follow from \eqref{eqn_lemma_3}, and $\stackrel{(5)}{=}$ follows from Lemma \ref{Lemma_N_properties}, \eqref{eqn_lemma_3} and the result on the expectation of the product of four Gaussian random variables \cite{Bar1971}. From \eqref{eqn_noise_var_terms}, we can then upper-bound the noise power as: 
\begin{eqnarray} \label{eqn_Z_pow}
\mathbb{E}\{{|Z_k|}^2\} \leq M_{\rm rx} \beta_{\rm max} \mathrm{N}_0 \big[ E^{(\rm r)} + |\hat{\mathcal{G}}| E^{(\rm d)} \big] + M_{\rm rx} |\hat{\mathcal{G}}| \mathrm{N}_0^2,
\end{eqnarray}
where we use the fact that $|\beta_{\dot{k},\dot{k}}| \leq \beta_{\rm max}$, $\sum_{\dot{k} \in \hat{\mathcal{G}}} \big[ \Delta_{\dot{k}+k, \dot{k}+k} + \Delta_{\dot{k}, \dot{k}} \big] \leq 1$ for $k \in \mathcal{K} \setminus \mathcal{G}$ (as $\hat{g} \leq g/2$) and $\sum_{\bar{k} \in \mathcal{K} \setminus \mathcal{G}} \Delta_{k+\dot{k}-\bar{k},k+\dot{k}-\bar{k}} \leq 1$, from \eqref{eqn_lemma_1}. 

\section{Performance Analysis} \label{sec_perf_anal}
From \eqref{eqn_r_k}--\eqref{eqn_Z_k}, the \emph{effective} single-input-single-output (SISO) channel between the $k$-th sub-carrier input and corresponding output can be expressed for $k \in \mathcal{K} \setminus \mathcal{G}$ as:
\begin{eqnarray} \label{eqn_effective_chan}
Y_k = M_{\rm rx} \beta_{0,k} \sqrt{E^{(\rm r)}} \Big[ \sum_{\dot{k} \in \hat{\mathcal{G}}} {|\Omega[\dot{k}]|}^2 \Big] x_{k} + I_k + Z_k, 
\end{eqnarray}
where $I_k$ and $Z_k$ are analyzed in Sections \ref{subsec_interf_anal} and \ref{subsec_noise_anal}, respectively. 
As is evident from \eqref{eqn_effective_chan}, the signal component suffers from two kinds of fading: (i) a frequency-selective and channel dependent slow fading component represented by $\beta_{0,k}$ and (ii) a frequency-flat and phase-noise dependent fast fading component, represented by $\sum_{\dot{k} \in \hat{\mathcal{G}}} {|\Omega[\dot{k}]|}^2$. The estimation of these fading coefficients is discussed later in this section. In this paper, we consider the simple demodulation approach where $x_k$ is estimated only from $Y_k$, and the $I_k, Z_k$ are treated as noise.\footnote{The estimation of $x_k$ from multiple OFDM sub-carriers outputs shall be explored in future work.} For this demodulation approach, a lower bound to the signal-to-interference-plus-noise ratio (SINR) can be obtained from \eqref{eqn_S_pow}, \eqref{eqn_I_pow_1}, Remark \ref{rem_int_approx} and \eqref{eqn_Z_pow}, as shown in \eqref{eqn_SINR} at the top of next page,
\begin{figure*}[h]
\begin{eqnarray}
\gamma^{\rm LB}_k(\boldsymbol{\beta}) \triangleq \frac{ M_{\rm rx} {|\beta_{0,k}|}^2 E^{(\rm r)} E^{(\rm d)} {\mu(0,\hat{g})}^2 }{ M_{\rm rx} {|\beta_{\rm max}|}^2 E^{(\rm r)} \Big[ E^{(\rm d)} \mu(0,\hat{g}) \big(1 \!-\! \mu(0,\hat{g})\big) \!+\! E^{(\rm r)} \tilde{\mu}(k,\hat{g}) \Big] \!+\! \beta_{\rm max} \mathrm{N}_0 \Big[ E^{(\rm r)} \!+\! |\hat{\mathcal{G}}| E^{(\rm d)} \Big] \!+\! |\hat{\mathcal{G}}| \mathrm{N}_0^2 }, \label{eqn_SINR}
\end{eqnarray}
\end{figure*}
where $\boldsymbol{\beta} \triangleq \{ \beta_{0,k} | \mathcal{K}\}$ and we use the fact that $\mathbb{E}\{{|I_k + Z_k|}^2\} = \mathbb{E}\{{|I_k|}^2\}+\mathbb{E}\{{|Z_k|}^2\}$.\footnote{Since $Z_k$ is the noise experienced while estimating $x_k$, it is inaccurate to take an expectation of ${|Z_k|}^2$ with respect to $x_k$, as in \eqref{eqn_noise_var_terms}. However the impact of this error is negligible when $K \gg 1$.} 
\begin{remark} \label{rem_array_orth}
If the RX array response vectors for the MPCs are mutually orthogonal i.e. ${\mathbf{a}_{\rm rx}(\ell_1)}^{\dag} \mathbf{a}_{\rm rx}(\ell_2) = M_{\rm rx} \delta_{\ell_1,\ell_2}^{\infty}$, then $\beta_{\dot{k}, \ddot{k}} = \sum_{\ell=0}^{L-1} {|\alpha_{\ell}|}^2 {|{\mathbf{a}_{\rm tx}(\ell) }^{\dag} \mathbf{t}|}^2 e^{{\rm j} 2 \pi (f_{\dot{k}}-f_{\ddot{k}}) \tau_{\ell}}$ and $\beta_{\rm max} = \bar{\beta}$, where we define $\bar{\beta} \triangleq \sum_{\ell=0}^{L-1} {|\alpha_{\ell}|}^2 \allowbreak {|{\mathbf{a}_{\rm tx}(\ell) }^{\dag} \mathbf{t}|}^2$. 
\end{remark}
The orthogonality of array response vectors is approximately satisfied if the MPCs are well separated and $M_{\rm rx} \gg L$ \cite{Ayach2012}. Additionally, while terms in $\beta_{\dot{k},\ddot{k}}$ combine incoherently, the resulting loss in $\gamma^{\rm LB}_k(\boldsymbol{\beta})$ is small for sparse wide-band channels with small $L$. Thus the CACE technique is very well suited for mm-wave massive MIMO channels where these conditions are typically satisfied. 
From Remark \ref{rem_array_orth}, note that even without explicit CE at the RX $\gamma^{\rm LB}_k(\boldsymbol{\beta})$ scales with $M_{\rm rx}$ in the low SNR regime, which is a desired characteristic. 
Though the ICI term also scales with $M_{\rm rx}$, its impact can be kept small in the desired SNR range by picking $\hat{g}$ such that $\mu(0,\hat{g}) \approx 1$. 
In a similar way, with perfect knowledge of the fading coefficients at the RX, an approximate lower bound to the ergodic capacity can be obtained as:\footnote{Here the ergodic capacity is computed assuming $\{\beta_{0,k} | k \in \mathcal{K}\}$ remain constant for infinite time but $\sum_{\dot{k} \in \hat{\mathcal{G}}} {|\Omega[\dot{k}]|}^2$ experiences many independent realizations. This capacity is representative of the throughput of practical codes that have a length spanning multiple OFDM symbols but smaller than coherence time of $\beta_{0,k}$ \cite{Foschini1998}.}
\begin{eqnarray}
C(\boldsymbol{\beta}) &\stackrel{(1)}{\geq} & \frac{1}{K}\sum_{k \in \mathcal{K} \setminus \mathcal{G}} \mathbb{E}_{\sum_{\dot{k} \in \hat{\mathcal{G}}} {|\Omega[\dot{k}]|}^2} \bigg\{ \log \bigg[ 1 + \nonumber \\
&& \qquad \frac{M_{\rm rx}^2 {|\beta_{0,k}|}^2 E^{(\rm r)} E^{(\rm d)} {\big[\sum_{\dot{k} \in \hat{\mathcal{G}}} {|\Omega[\dot{k}]|}^2 \big]}^2}{ \mathbb{E} \big\{ {|I_k|}^2 + {|Z_k|}^2 \Big| \sum_{\dot{k} \in \hat{\mathcal{G}}} {|\Omega[\dot{k}]|}^2 \big\} } \bigg] \bigg\} \nonumber \\
&\stackrel{(2)}{\approx} & \frac{1}{K}\sum_{k \in \mathcal{K} \setminus \mathcal{G}} \bigg( \log \big[ \mathbb{E} \{ {|I_k|}^2 \!\!+ {|Z_k|}^2\} \nonumber \\
&& + M_{\rm rx}^2 {|\beta_{0,k}|}^2 E^{(\rm r)} E^{(\rm d)} {\mu(\hat{g})}^2 \big] - \log[ \mathbb{E} \{ {|I_k|}^2 \!\!+ {|Z_k|}^2\}] \bigg)\nonumber \\
& \stackrel{(3)}{\geq} & \frac{1}{K}\sum_{k \in \mathcal{K} \setminus \mathcal{G}} \log \big[ 1 + \gamma^{\rm LB}_{k} (\boldsymbol{\beta}) \big] \triangleq C_{\rm approx}(\boldsymbol{\beta}), \label{eqn_erg_cap}
\end{eqnarray}
where $\stackrel{(1)}{\geq}$ is obtained by assuming $I_k, Z_k$ to be Gaussian distributed and using the expression for ergodic capacity \cite{Goldsmith1997}, $\stackrel{(2)}{\approx}$ follows by sending the outer expectation into the $\log(\cdot)$ functions and  $\stackrel{(3)}{\geq}$ follows from \eqref{eqn_S_pow}, \eqref{eqn_I_pow_1} and \eqref{eqn_Z_pow}. While $\stackrel{(2)}{\approx}$ is an approximation, it typically yields a lower bound since ${\rm Variance}\{\sum_{\dot{k} \in \hat{\mathcal{G}}} {|\Omega[\dot{k}]|}^2 \} \leq  \mu(0,\hat{g}) [1 - \mu(0,\hat{g})] \ll \mu(0,\hat{g})^2$ (from \eqref{eqn_lemma_1} and \cite{Bhatia2000}).

Note that for demodulating $x_k$'s and achieving the above SINR and capacity, the RX requires estimates of $\mathrm{N}_0$ and the SISO channel fading coefficients $\boldsymbol{\beta}$ and $\sum_{\dot{k} \in \hat{\mathcal{G}}} {|\Omega[\dot{k}]|}^2$. Since the RX has a good beamforming gain \eqref{eqn_SINR}, the channel parameters $\boldsymbol{\beta}, \mathrm{N}_0$ can be tracked accurately at the RX with a low estimation overhead using pilot symbols and blanked symbols. 
These values, along with phase-noise parameter $\sigma_{\theta}$, can further be fed back to the TX for rate and power allocation. Note that since these pilots are only used to estimate the SISO channel parameters and not the actual MIMO channel, the advantages of simplified CE are still applicable for a CACE based RX. On the other hand, the low variance albeit fast varying component $ \sum_{\dot{k} \in \hat{\mathcal{G}}} {|\Omega[\dot{k}]|}^2$ can be estimated for every symbol using the $0$-th sub-carrier output $Y_0$. It can be shown from \eqref{eqn_r_k} that $Y_0 = M_{\rm rx} \beta_{0,0} E^{(\rm r)} \Big[ \sum_{\dot{k} \in \hat{\mathcal{G}}} {|\Omega[\dot{k}]|}^2 \Big] + I_0 + M_{\rm rx} |\hat{\mathcal{G}}| \mathrm{N}_0 + Z_0$, where we have $\mathbb{E}\{{|I_0|}^2\} \leq \mathbb{E}\{{|I_k|}^2\}$ and $\mathbb{E}\{{|Z_0|}^2\} \leq 2 \mathbb{E}\{{|Z_k|}^2\}$ for any $k \in \mathcal{K}\setminus \mathcal{G}$.\footnote{While the derivations follow similar steps to those in Section \ref{sec_demod_anal}, the explicit proof is skipped for brevity.} Thus $\sum_{\dot{k} \in \hat{\mathcal{G}}} {|\Omega[\dot{k}]|}^2$ can be estimated from $Y_0$ with an SINR $\geq \frac{E^{\rm r} \gamma^{\rm LB}_k(\boldsymbol{\beta}) }{2 E_{(\rm d)}}$, which is usually a large value.

\subsection{Optimizing the system parameters} \label{subsec_parameter_choices}
In this section we find capacity maximizing values of the system parameters $g, E^{(\rm r)}$ and $\hat{g}$. From \eqref{eqn_erg_cap}, note that the approximate ergodic capacity $C_{\rm approx}(\boldsymbol{\beta})$ is a decreasing function of $g$ for $g \geq 2\hat{g}$. Thus a $C_{\rm approx}(\boldsymbol{\beta})$ maximizing choice of $g$ should satisfy $g=2\hat{g}$. To find a near-optimal values of $E^{(\rm r)}$ and $\hat{g}$, we further lower bound $C_{\rm approx}(\boldsymbol{\beta})$ using \eqref{eqn_erg_cap} and \eqref{eqn_SINR}, as: 
\begin{subequations}
\begin{flalign}
& C_{\rm approx}(\boldsymbol{\beta}) \stackrel{(1)}{\geq} \frac{1}{K}\sum_{k \in \mathcal{K} \setminus \mathcal{G}} \!\! \log({|\beta_{0,k}|}^2) + \frac{K-|\mathcal{G}|}{K}\log[\Xi(\boldsymbol{\beta})], \!\!\!\! & \label{eqn_erg_cap_aprox2}
\end{flalign}
\begin{figure*}
\begin{eqnarray}
\Xi(\boldsymbol{\beta}) & \stackrel{(2)}{=} & \frac{ M_{\rm rx} E^{(\rm r)} E^{(\rm d)} {\mu(0,\hat{g})}^2 }{ M_{\rm rx} {|\beta_{\rm max}|}^2 E^{(\rm r)} \big( \frac{E_{\rm s}}{K-|\mathcal{G}|} \big) \mu(0,\hat{g}) \big[1 \!-\! \mu(0,\hat{g})\big] \!+\! \beta_{\rm max}\mathrm{N}_0 \big( E^{(\rm r)} \!+\! |\hat{\mathcal{G}}| E^{(\rm d)} \big) \!+\! |\hat{\mathcal{G}}| \mathrm{N}_0^2 },  \label{eqn_SINR_LB}
\end{eqnarray}
\end{figure*}
\end{subequations}
\noindent where $\Xi(\boldsymbol{\beta})$ is as given in \eqref{eqn_SINR_LB} at the top of this page, $\stackrel{(1)}{\geq}$ follows from the fact that $\log(1+ \gamma^{\rm LB}_{k} (\boldsymbol{\beta})) \geq \log(\gamma^{\rm LB}_{k} (\boldsymbol{\beta}))$ and by taking the summation over $k$ in \eqref{eqn_erg_cap} into the denominator of the logarithm; and $\stackrel{(2)}{=}$ in \eqref{eqn_SINR_LB} follows from the fact that $\sum_{k \in \mathcal{K} \setminus \mathcal{G}} \tilde{\mu}(k,\hat{g}) \leq \sum_{k \in \mathcal{K} \setminus \hat{\mathcal{G}}} \mu(0, \hat{g}) \Delta_{k,k}$ and $E^{(\rm d)} ({K-|\mathcal{G}|}) + E^{(\rm r)} = E_{\rm s}$. 
It can be verified that the numerator of $\Xi(\boldsymbol{\beta})$ is a differentiable, strictly concave function of $E^{(\rm r)}$, while the denominator is a positive, affine function of $E^{(\rm r)}$. Thus $\Xi(\boldsymbol{\beta})$ is a strictly pseudo-concave function of $E^{(\rm r)}$ \cite{Schaible1983}, and the $C_{\rm approx}(\boldsymbol{\beta})$ maximizing power allocation can be obtained by setting $\frac{ {\rm d} \Xi(\boldsymbol{\beta})}{{\rm d} E^{(\rm r)}} = 0$ as:
\begin{eqnarray} \label{eqn_Er_opt}
E^{(\rm r)}_{\rm opt} = E_{\rm s} \frac{\sqrt{R^2 + QR}-R}{Q}
\end{eqnarray}
where $Q = M_{\rm rx} {|\beta_{\rm max}|}^2 [1-\mu(0,\hat{g})] \mu(0,\hat{g}) E_{\rm s} + \beta_{\rm max} \mathrm{N}_0 (K-|\hat{\mathcal{G}}|-|\mathcal{G}|)$ and $R = \mathrm{N}_0 |\hat{\mathcal{G}}| \big[ \beta_{\rm max}  + \mathrm{N}_0 (K-|\mathcal{G}|) / E_{\rm s} \big]$. 
As evident from \eqref{eqn_SINR_LB}, $\hat{g}$ offers a trade-off between the phase-noise induced ICI and the channel noise accumulation. While finding a closed form expression for \eqref{eqn_erg_cap_aprox2} maximizing $\hat{g}$ is intractable, it can be computed numerically by performing a simple line search over $1 \leq \hat{g} \leq \min\{K_1,K_2\}/2$, with $g = 2 \hat{g}$ and $E^{(\rm r)}$ as given by \eqref{eqn_Er_opt}. 

\section{Initial Access, TX beamforming and uplink beamforming} \label{sec_IA_aCSI_at_BS}
In this section we briefly discuss stages (i) and (ii) of downlink transmission (see Section \ref{sec_chan_model}), and uplink TX beamforming for CACE aided UEs. 
In the suggested IA protocol for stage (i), the BS performs beam sweeping along different angular directions, possibly with different beam widths, similar to the approach of 3GPP New Radio (NR). For each TX beam, the BS transmits primary (PSS) and secondary synchronization sequences (SSS) with the reference signal, in a form similar to \eqref{eqn_tx_signal}. The UEs use CACE aided RX beamforming, and initiate uplink random access to the BS upon successfully detecting a PSS/SSS. As shall be shown in Section \ref{sec_sim_results}, the SINR expression \eqref{eqn_SINR} is resilient to frequency mismatches between TX and RX oscillators, and thus is also applicable for the PSS/SSSs where frequency synchronization may not exist. Since angular beam-sweeping is only performed at the BS, the IA latency does not scale with $M_{\rm rx}$ and yet the PSS/SSS symbols can exploit the RX beamforming gain, thus improving cell discovery radius and/or reducing IA overhead. This is in contrast to digital CE at the UE, which would require sweeping through many RX beam directions for each TX direction, necessitating several repetitions of the PSS/SSS for each TX beam. 
During downlink stage (ii), note that scheduling of UEs, designing TX beamformer $\mathbf{t}$ and allocation of power requires knowledge of $\{|\alpha_{\ell}|, \mathbf{a}_{\rm tx}(\ell)\}$ for all the UEs. 
Such rCSI can be acquired at the BS either by downlink CE with rCSI feedback from the UEs or by uplink CE. The protocol for downlink CE with feedback is similar to the IA protocol, with the BS transmitting pilot symbols instead of PSS/SSS along different candidate $\mathbf{t}$'s. Uplink CE can be performed by transmitting orthogonal pilots from the UEs omni-directionally, and using any of the digital CE algorithms from Section \ref{sec_intro} at the BS. Note that CACE cannot be used at the BS since the pilots from multiple UEs need to be separated via digital processing. 

Note that the phase shifts used for RX beamforming at a CACE aided UE in downlink, can also be used for transmit beamforming in the uplink. However since the reference signal is not available at the UE during uplink transmission in time division duplexing systems, a mechanism for locking these phase shift values from a previous downlink transmission stage is required (similar to \cite{Ratnam_PACE}). In contrast, frequency division duplexing can avoid such a mechanism due to continuous availability of the downlink reference, and consequently $\hat{\mathbf{s}}_{\rm rx,BB}(t)$. 

\section{Simulation Results} \label{sec_sim_results}
For the simulations, we consider a single cell scenario with a $\lambda/2$-spaced $32 \times 8$ ($M_{\rm tx} = 256$) antenna BS and one representative UE: with a $\lambda/2$-spaced $16 \times 4$ ($M_{\rm rx} = 64$) antenna array, one down-conversion chain, using CACE aided beamforming and having perfect timing synchronization to the BS. The BS has apriori rCSI and transmits one spatial OFDM data stream with $T_{\rm s}=1 \mu$s, $K_1=K_2+1=512$ and $f_{\rm c} = 30$ GHz along the strongest MPC, i.e., $\mathbf{t} = \mathbf{a}_{\rm tx}(\bar{l})$ for $\bar{l} = \rm{argmax}_{\ell}\{|\alpha_{\ell}|\}$. 
The UE oscillator experiences phase-noise with variance $\sigma_{\theta}^2=1/\sqrt{T_{\rm s}}$ known both to the BS and UE. The UE also has perfect knowledge of $\boldsymbol{\beta}, \mathrm{N}_0$ and $\sum_{\dot{k} \in \hat{\mathcal{G}}} {|\Omega[\dot{k}]|}^2$. For convenience, we shall use $\bar{\beta} E_{\rm s}/K \mathrm{N}_0$ to quantify the simulation SNR, which reflects the mean SNR at any RX antenna without RX beamforming gain (see Remark \ref{rem_array_orth}). 

For testing the validity of the analytical results, we first consider a sparse channel matrix $\mathbf{H}(t)$ with $L=3$, $\hat{\tau}_{\ell}= \{0, 20, 40\}$ns, angles of arrival $\psi^{\rm rx}_{\rm azi} = \{0, \pi/6, -\pi/6\}$, $\psi^{\rm rx}_{\rm ele} = \{0.45\pi, \pi/2, \pi/2\}$ and normalized amplitudes $\frac{\alpha_{\ell} \mathbf{a}_{\rm tx}(\ell)^{\dag} \mathbf{t}}{\sqrt{\bar{\beta}}} = \{\sqrt{0.6}, -\sqrt{0.3}, \sqrt{0.1}\}$. The UE uses $\hat{g} = g/2 = 10$ and $E^{(\rm r)}, E^{(\rm d)}$ from \eqref{eqn_Er_opt}. For this model, the symbol error rates (SERs) for the sub-carriers, obtained by Monte-Carlo simulations, are compared to the analytical SERs for a Gaussian channel with SINR given by \eqref{eqn_SINR} (with/without Remark \ref{rem_int_approx}) in Fig.~\ref{Fig_SER_compare}. For the Monte-Carlo results, we use truncated sinc filters: ${\rm LPF}_{\hat{g}}(t) = \sin(2 \pi \hat{g} t/T_{\rm s}) \big/ (\pi t)$ for $|t| \leq 2T_{\rm s}/\hat{g}$. 
As observed from the results and mentioned in Section \ref{subsec_interf_anal}, the use of Remark \ref{rem_int_approx} in \eqref{eqn_SINR} provides a tight SINR bound even for small $|k|$.  
We also observe that the SER for $k = 22 \ (\approx \hat{g})$ is high due to the ICI caused from the high power reference signal. However this ICI diminishes very quickly with $k$ due to phase noise suppression, as evident from the SER for $k=-40$. While the mean RX oscillator frequency was assumed to be perfectly matched to the TX oscillator for the analytical results (see Section \ref{sec_chan_model}), we also plot in Fig. \ref{Fig_SER_compare} the case with a $5$ MHz frequency mismatch. Results show a negligible degradation in performance, suggesting that the CACE design is resilient to oscillator frequency mismatches that are smaller than the cut-off frequency of ${\rm LPF}_{\hat{g}}$. For computational tractability, and due to the accuracy of the bounds in Fig.~\ref{Fig_SER_compare}, we shall henceforth use \eqref{eqn_SINR} and \eqref{eqn_erg_cap} to quantify the performance of CACE for the remaining results. 
\begin{figure}[!htb]
\centering
\includegraphics[width= 0.45\textwidth]{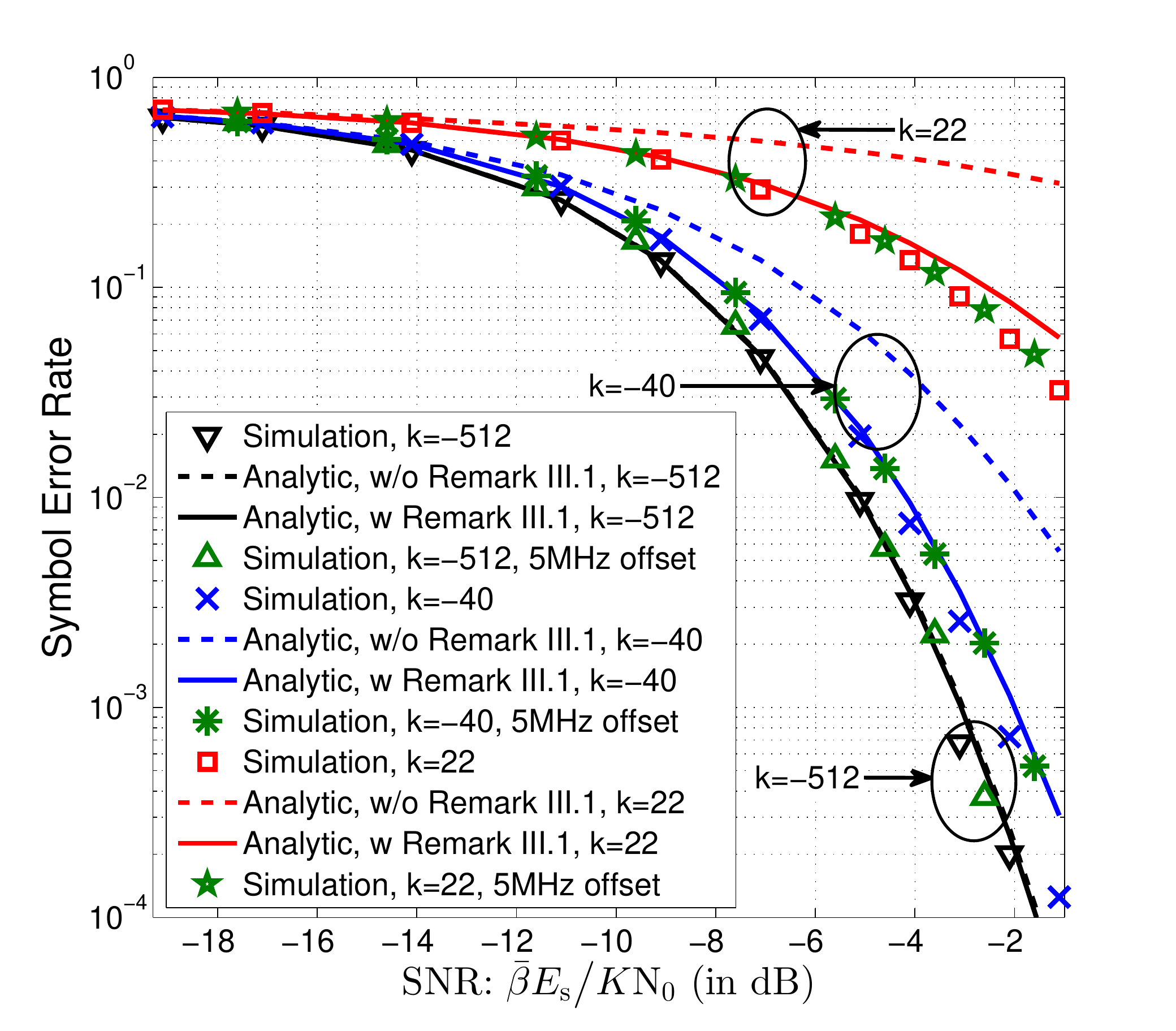}
\caption{Comparison of analytical SERs (from \eqref{eqn_SINR} with/without Remark \ref{rem_int_approx}) to simulated results, for different sub-carriers of a CACE receiver with Quadrature Phase Shift Keying. Simulations consider $\sigma_{\theta}^2=1/T_{\rm s}$ ($-93$dBc at $10$MHz offset), mean RX oscillator frequency of $f_{\rm c}$ and $f_{\rm c}+5$MHz, and are obtained by averaging over $10^6$ realizations.}
\label{Fig_SER_compare}
\end{figure}

We next plot $C_{\rm approx}(\boldsymbol{\beta})$ from \eqref{eqn_erg_cap} as a function of $\hat{g}$ in Fig.~\ref{Fig_pow_alloc_compare}, with (a) $C_{\rm approx}(\boldsymbol{\beta})$ maximizing $E^{(\rm r)}$ (obtained by exhaustive search over $0 \leq E^{(\rm r)} \leq E_{\rm s}$) and (b) $E^{(\rm r)}$ chosen from \eqref{eqn_Er_opt}, respectively. As observed from the results, the curves are very close, suggesting the accuracy of the power allocation in \eqref{eqn_Er_opt}. Fig.~\ref{Fig_pow_alloc_compare} also demonstrates the trade-off characterized by $\hat{g}$: where ICI degrades the performance for small $\hat{g}$ and the noise accumulation, spectral efficiency reduction degrade performance for large $\hat{g}$. 
We also note that the optimal $\hat{g}$ increases with SNR. 
\begin{figure}[!htb]
\centering
\includegraphics[width= 0.45\textwidth]{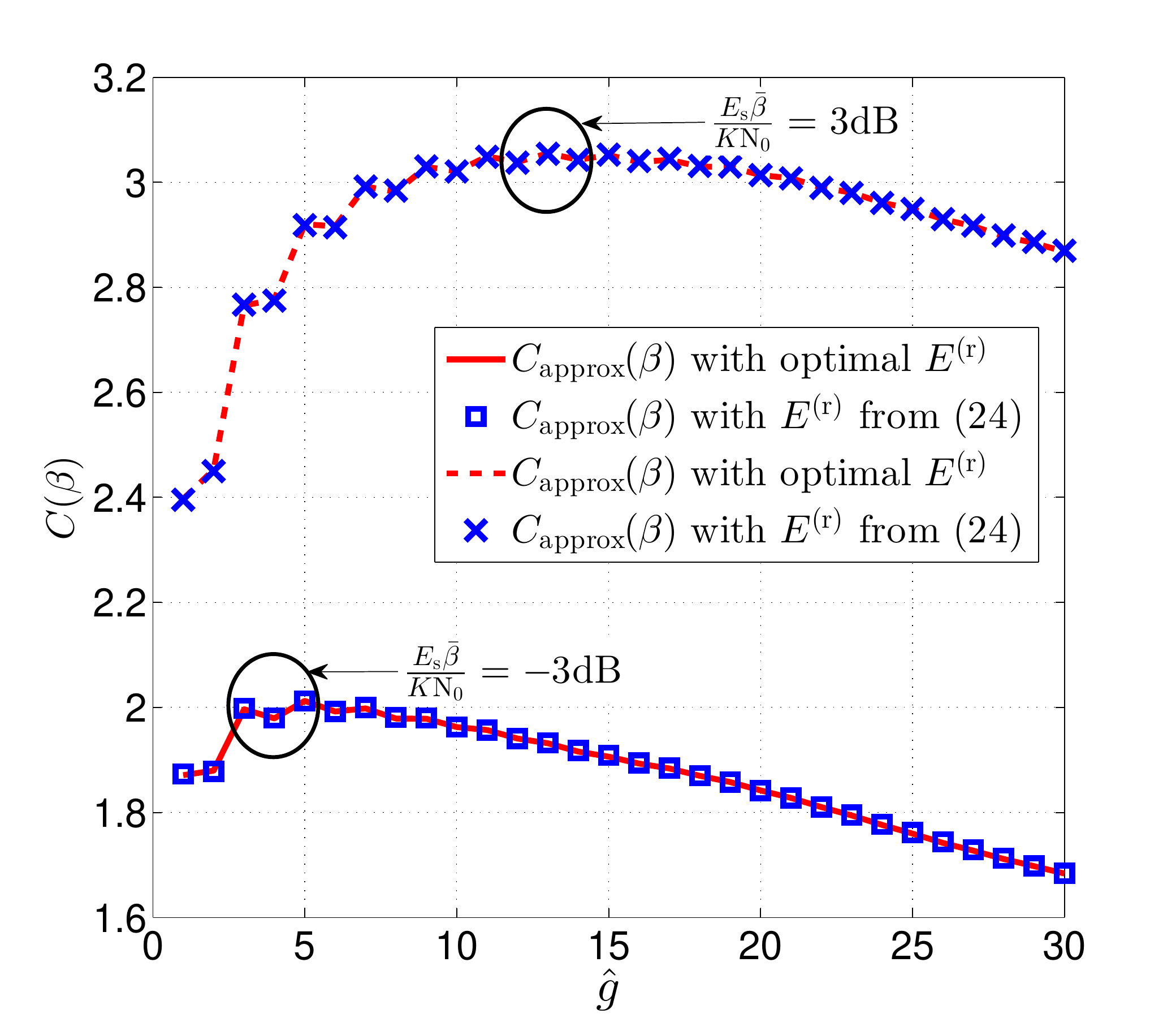}
\caption{Comparison of approximate capacity (from \eqref{eqn_erg_cap}) versus $\hat{g}$, with optimal choice of $E^{(\rm r)}$ and $E^{(\rm r)}$ chosen via \eqref{eqn_Er_opt}, respectively, for $\sigma_{\theta}^2=1/T_{\rm s}$ and $\bar{\beta} E_{\rm s}/K \mathrm{N}_0 = -3, 3$dB.}
\label{Fig_pow_alloc_compare}
\end{figure}

Fig.~\ref{Fig_cap_compare_sparse} compares the achievable throughput (excluding CE overhead) for beamforming with digital CE and different ACE schemes: CACE, PACE \cite{Ratnam_PACE}, MA-FSR \cite{Ratnam_Globecom2018}, respectively for the sparse channel defined above. For digital CE, the RX beamformer is aligned with the largest eigenvector of the effective RX correlation matrix $\mathbf{R}_{\rm rx}(\mathbf{t}) = \frac{1}{K} \sum_{k \in \mathcal{K}}\boldsymbol{\mathcal{H}}(f_k) \mathbf{t} \mathbf{t}^{\dag} {\boldsymbol{\mathcal{H}}(f_k)}^{\dag}$ \cite{Sudarshan2006}, which in turn is either (a) known apriori at the BS or (b) is estimated by nested array based sampling \cite{Pal2010}. To decouple the loss in beamforming gain due to CE errors from loss due to phase-noise, we assume $\sigma_{\theta} \approx 0$. As is evident from Fig.~\ref{Fig_cap_compare_sparse}, PACE and CACE suffer only a $\leq 2$dB beamforming loss in compared to digital CE in sparse channels and above a threshold SNR. While CACE performs marginally worse than PACE at high SNR due to power wastage on a continuous reference, unlike PACE it does not suffer from PLL based carrier recovery losses at low SNR. 
While MA-FSR performs poorly due to low bandwidth efficiency, it requires much simpler hardware then all other schemes. To demonstrate the phase-noise suppressing capability of CACE (and MA-FSR), we also plot the throughput of CACE (with optimal $\hat{g}$) and digital CE, with $\sigma_{\theta}^2 = 1/T_{\rm s}$ and without any additional phase-noise mitigation. As is evident from the results, both CACE and MA-FSR aid in mitigating oscillator phase-noise in addition to enabling RX beamforming. 
To study the impact of more realistic channels and number of MPCs, we also consider a rich scattering stochastic channel in Fig.~\ref{Fig_cap_compare_dense}, having $L/10$ resolvable MPCs and $10$ sub-paths per resolvable MPC. All channel parameters are generated according to the 3GPP TR38.900 Rel 14 channel model (UMi NLoS scenario) \cite{TR38900_chanmodel}, with the resolvable MPCs and sub-paths modeled as clusters and rays, respectively. However to model the sub-paths of each MPC as unresolvable, we use an intra-cluster delay spread of $1 ns$ and an intra-cluster angle spread of $\pi/50$ (for all elevation, azimuth, arrival and departure).
As observed, the performance of ACE schemes degrades slightly faster with $L$ than of digital CE due to the incoherent combining of the MPC contributions in $\beta_{0, k}$ (see Remark \ref{rem_array_orth}). 
%
\begin{figure}[!htb]
\centering
\subfloat[Sparse channel]{\includegraphics[width= 0.45\textwidth]{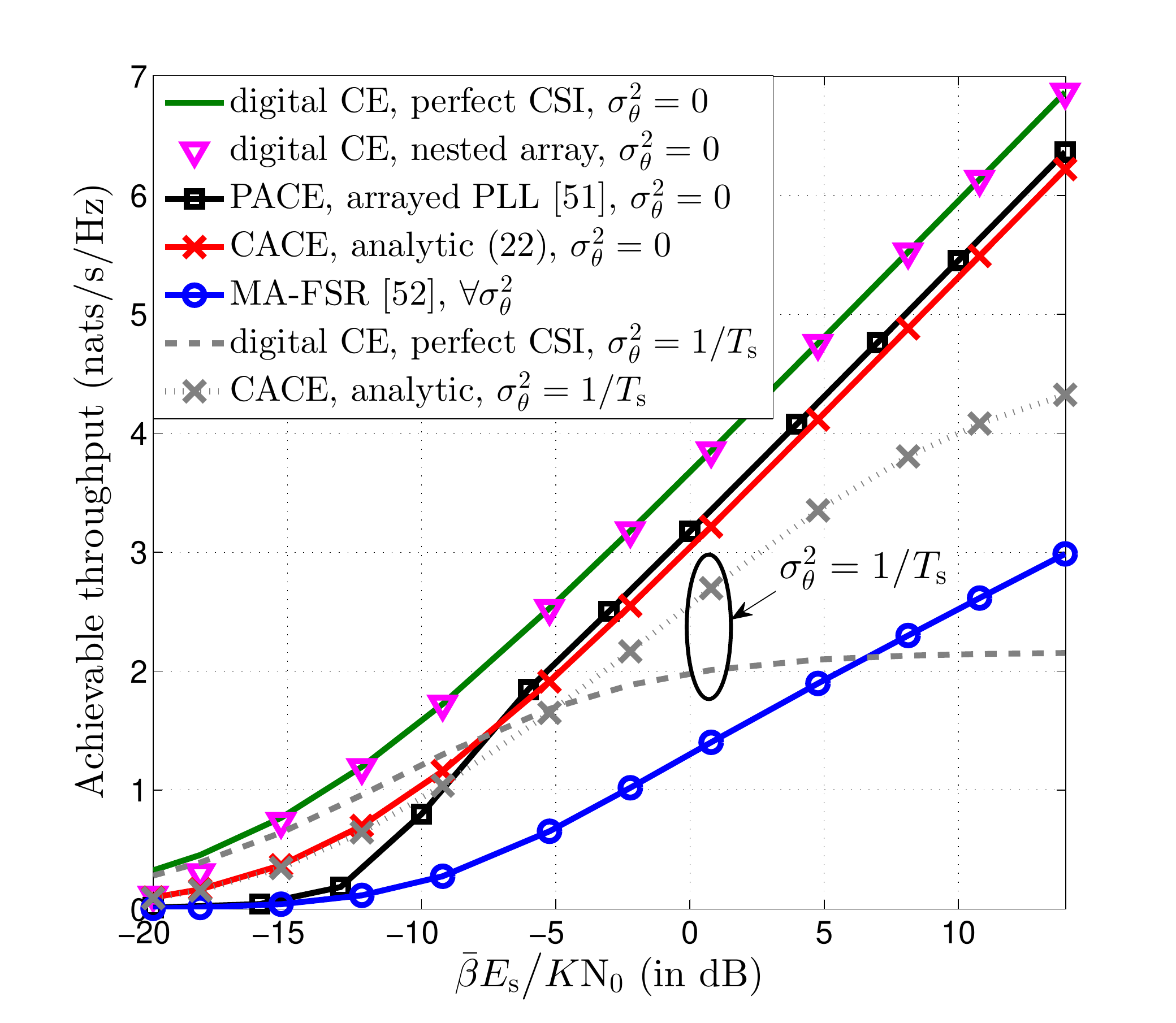} \label{Fig_cap_compare_sparse}} \\
\subfloat[Dense stochastic channel]{\includegraphics[width= 0.45\textwidth]{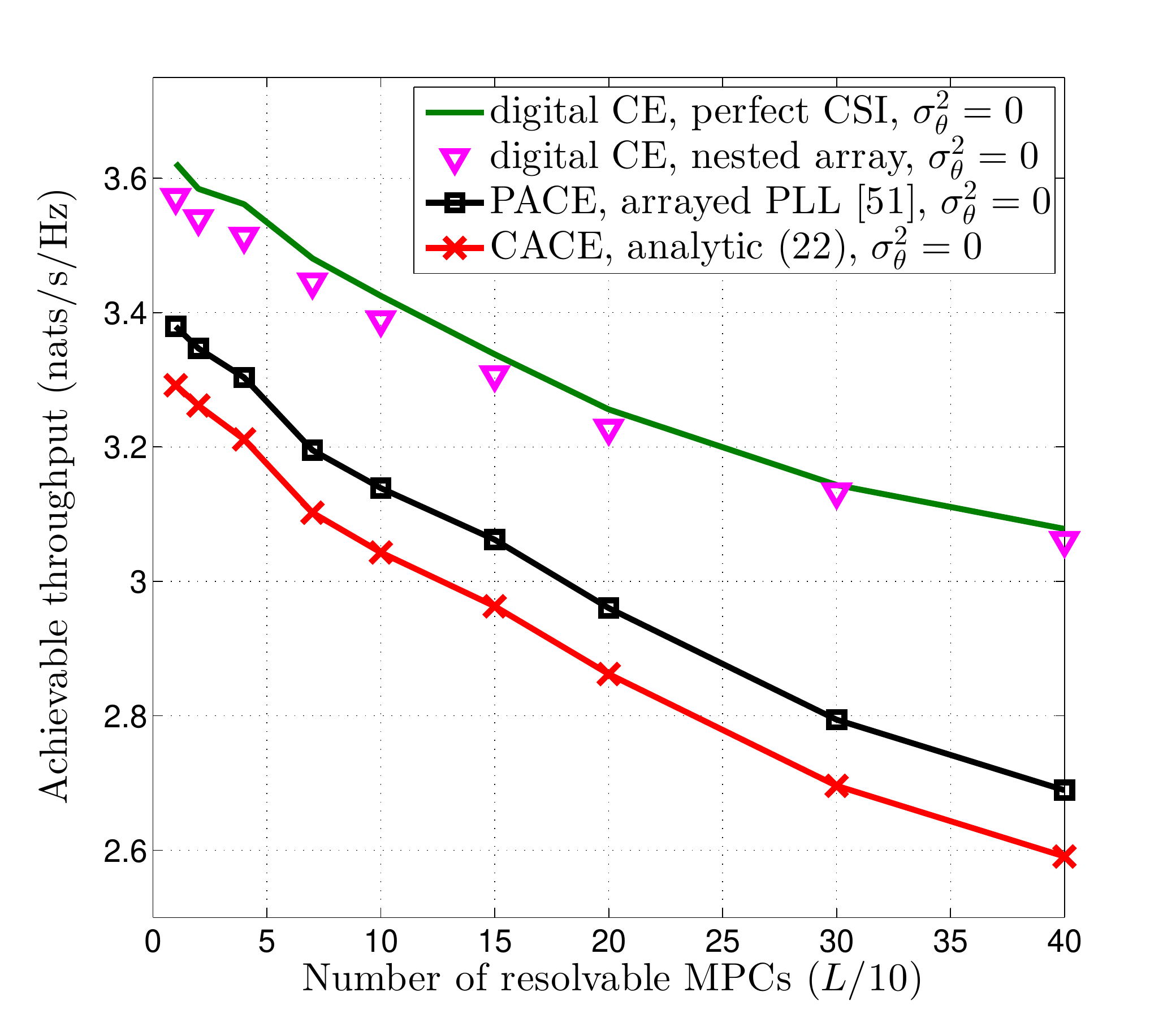} \label{Fig_cap_compare_dense}}
\caption{Throughput of ACE schemes (PACE, CACE, MA-FSR) and of digital CE with either perfect rCSI or nested array sampling, versus SNR and $L$. For PACE, the arrayed PLL from \cite{Ratnam_PACE} is used with identical parameters, and RX beamformer design phase lasts $6$ symbols. For nested array sampling, horizontal($4,4$) and vertical $(2,2)$ nested arrays are used. For Fig.~\ref{Fig_cap_compare_dense}, we use $\bar{\beta} E_{\rm s}/K \mathrm{N}_0 = 1$.}
\label{Fig_cap_compare}
\end{figure}

Note that the results in Fig.~\ref{Fig_cap_compare} do not include the CE overhead for PACE and digital CE. While nested array digital CE requires $21$ dedicated pilot symbols ($\approx 2\sqrt{M_{\rm rx}}$) for updating RX beamformer, PACE requires $6$ symbols (${\rm O}(1)$) and CACE, MA-FSR only require a continuous reference tone. The corresponding overhead reduction can be significant when downlink CE with rCSI feedback is used for rCSI acquisition at the BS (see Section \ref{sec_IA_aCSI_at_BS}). For example, with exhaustive beam-scanning \cite{Jeong2015} at the BS and an rCSI coherence time of $10$ms, the BS rCSI acquisition overhead reduces from $40$\% for nested array digital CE to $11$\% for PACE and $\approx |\mathcal{G}|/K < 5\%$ for CACE. 

\section{Conclusions} \label{sec_conclusions}
This paper proposes a novel CE technique called CACE for designing the RX beamformer in massive MIMO systems. CACE enables both RX beamforming and phase-noise cancellation at very low CE overhead. 
The performance analysis suggests that in sparse channels and for low-pass filter bandwidth parameter $\hat{g} \gg 1$, the SINR with CACE scales linearly with the number of receive antennas $M_{\rm rx}$. The analysis and simulations also show that $\hat{g}$ yields a trade-off between phase-noise induced ICI and noise accumulation. Simulations suggest that CACE suffers only a small degradation in beamforming gain in comparison to digital CE based beamforming in sparse channels, and is resilient to TX-RX oscillator frequency mismatch. In comparison to other ACE schemes, CACE performs marginally worse than PACE at high SNR but performs much better at lower SNR. It also performs much better than MA-FSR, albeit at a higher RX hardware complexity. Finally, CACE also provides phase-noise suppression unlike most other CE schemes. The CE overhead reduction with CACE is significant, especially when downlink CE with feedback is required. The IA latency reduction with CACE aided beamforming is also discussed. While base-band phase shifters are sufficient for a CACE based RX unlike in conventional analog beamforming, $2M_{\rm rx}$ mixers may be required for the base-band conversion at the RX, thus adding to the hardware cost.

\begin{appendices}
\section{} \label{appdix1}
\begin{proof}[Proof of Lemma \ref{Lemma_PN_properties}]
Note that from the definition of $\Omega[k]$, we have $e^{- {\rm j}\theta[n]} \stackrel{\mathcal{F}}{\longrightarrow} \Omega[k]$ and $e^{{\rm j}\theta[n]} \stackrel{\mathcal{F}}{\longrightarrow} \Omega^{*}[-k]$, where $\mathcal{F}$ represents the nDFT Operation. Then using convolution property of the nDFT, we have:
\begin{eqnarray}
e^{- {\rm j}\theta[n]} e^{{\rm j}\theta[n]} \stackrel{\mathcal{F}}{\longrightarrow} \sum_{a \in \mathcal{K}} \Omega[a] \Omega^{*}[a+k] \nonumber \\
\Rightarrow 1 \stackrel{\mathcal{F}}{\longrightarrow} \sum_{a \in \mathcal{K}} \Omega[a] \Omega^{*}[a+k] \nonumber \\
\Rightarrow \delta^{K}_{0,k} = \sum_{a \in \mathcal{K}} \Omega[a] \Omega^{*}[a+k] \nonumber
\end{eqnarray}
which proves property \eqref{eqn_lemma_1}. Property \eqref{eqn_lemma_3} can be obtained as follows:
\begin{flalign}
& \Delta_{k_1,k_2} \triangleq \mathbb{E}\{\Omega[k_1] {\Omega[k_2]}^{*}\} &\nonumber \\
& \qquad = \frac{1}{K^2} \sum_{\dot{n}, \ddot{n} = 0}^{K-1} \mathbb{E}\{ e^{-{\rm j}[\theta[\dot{n}]-\theta[\ddot{n}]} \} e^{-{\rm j}2 \pi \frac{[k_1 \dot{n} - k_2 \ddot{n}]}{K}} & \nonumber \\
& \qquad \stackrel{(1)}{=} \frac{1}{K^2} \sum_{\dot{n}, \ddot{n} = 0}^{K-1} e^{-\frac{\sigma_{\theta}^2 |\dot{n}-\ddot{n}| T_{\rm s}}{2K}} e^{-{\rm j}2 \pi \frac{[k_1 \dot{n} - k_2 \ddot{n}]}{K}} & \nonumber \\
& \qquad \stackrel{(2)}{=} \frac{1}{K^2} \sum_{\ddot{n} = 0}^{K-1} \sum_{u = -\ddot{n}}^{K-1-\ddot{n}} e^{-\frac{\sigma_{\theta}^2 |u| T_{\rm s}}{2K}} e^{-{\rm j}2 \pi \frac{[k_1 u + (k_1- k_2) \ddot{n}]}{K}} & \nonumber \\
& \qquad \stackrel{(3)}{\approx} \frac{1}{K^2} \sum_{\ddot{n} = 0}^{K-1} \sum_{u = -K/2}^{ K/2-1} e^{-\frac{\sigma_{\theta}^2 |u| T_{\rm s}}{2K}} e^{-{\rm j}2 \pi \frac{[k_1 u + (k_1- k_2) \ddot{n}]}{K}} & \nonumber \\
& \qquad \stackrel{(4)}{=} \frac{\delta_{k_1,k_2}^{K}}{K} \bigg[ \frac{1 - e^{-(\frac{\sigma_{\theta}^2 T_{\rm s} - {\rm j} 4 \pi k_1}{4} )}}{e^{\frac{\sigma_{\theta}^2 T_{\rm s} - {\rm j}4\pi k_1}{2K}} - 1} + \frac{1 - e^{-(\frac{\sigma_{\theta}^2 T_{\rm s} + {\rm j}4 \pi k_1}{4} )}}{1 - e^{-\frac{\sigma_{\theta}^2 T_{\rm s} + {\rm j}4 \pi k_1}{2K}}} \bigg] & \label{eqn_lemma1_proof1}
\end{flalign}
where $\stackrel{(1)}{=}$ follows by using the expression for the characteristic function of the Gaussian random variable $\theta[\dot{n}] - \theta[\ddot{n}]$; $\stackrel{(2)}{=}$ follows by defining $u = \dot{n} - \ddot{n}$ and $\stackrel{(3)}{\approx}$ follows by changing the inner summation limits which is accurate for $\sigma_{\theta}^2 T_{\rm s} \gg 1$ and $\stackrel{(4)}{=}$ follows from the expression for the sum of a geometric series. 
\end{proof}

\section{} \label{appdix3}
\begin{proof}[Proof of Lemma \ref{Lemma_N_properties}]
Note that each component of $\tilde{\mathbf{w}}(t)$ is independent and identically distributed as a circularly symmetric Gaussian random process. Hence its nDFT coefficients, obtained as $\mathbf{W}[k] = \frac{1}{K} \sum_{n=0}^{K-1} \tilde{\mathbf{w}}(n T_{\rm s}/K) e^{-{\rm j}2 \pi \frac{k n}{K}}$ are also jointly Gaussian and circularly symmetric. For these coefficients at RX antennas $a,b$ we obtain:
\begin{flalign}
& \mathbb{E}\{ W_a[k_1] W_b[k_2] \} & \nonumber \\
& \quad = \frac{1}{K^2}\sum_{n_1,n_2 = 1}^{K} \mathbb{E}\{ \tilde{w}_{a}(n_1 T_{\rm s}/K) \tilde{w}_{b} (n_2 T_{\rm s}/K)\} e^{-{\rm j}2 \pi \frac{k_1 n_1 + k_2 n_2}{K}} & \nonumber \\
& \quad = 0 & \\
& \mathbb{E}\{ W_a[k_1] W^{*}_b[k_2] \} & \nonumber \\
& \quad = \frac{1}{K^2}\sum_{n_1,n_2 = 1}^{K} \mathbb{E}\{ \tilde{w}_{a}(n_1 T_{\rm s}/K) \tilde{w}_{b}(n_2 T_{\rm s}/K)^{*}\} e^{-{\rm j}2 \pi \frac{k_1 n_1 - k_2 n_2}{K}} & \nonumber \\
& \quad = \frac{\delta^{\infty}_{a,b}}{K^2} \sum_{n_1,n_2 = 1}^{K} R_{\tilde{w}}\big([n_1-n_2] T_{\rm s}/K\big) e^{-{\rm j}2 \pi \frac{k_1 n_1 - k_2 n_2}{K}} & \nonumber \\
& \quad = \delta^{\infty}_{a,b} \delta^{K}_{k_1,k_2} {\mathrm{N}_0} \big/ {T_{\rm s}} & \nonumber
\end{flalign}
where we use the auto-correlation function of the channel noise at any RX antenna as: $R_{\tilde{w}}(t) = \mathrm{N}_0 \sin( \pi K t/T_{\rm s}) \exp{\{-{\rm j} \pi (K_1-K_2)t/T_{\rm s}\}} \Big/{\pi t}$. 
\end{proof}

\section{} \label{appdixOU}
Here we model the RX phase-noise $\theta(t)$ as a zero mean Ornstein-Ulhenbeck (OU) process \cite{Doob1942}, which is representative of the output of a type-1 phase-locked loop with a linear phase detector \cite{Viterbi_book, Mehrotra2002, Petrovic2007}. For such a model, $\theta(t)$ satisfies:
\begin{eqnarray}
\frac{{\rm d}\theta(t)}{{\rm d}t} = - \eta_{\theta} \theta(t) + \sigma_{\theta} w_{\rm \theta}(t)
\end{eqnarray}
where, $w_{\theta}(t)$ is a standard real white Gaussian process, and $\eta_{\theta}, \sigma_{\theta}$ are system parameters. 
From \eqref{eqn_PN_model} it can be shown that $\theta(t)$ is a stationary Gaussian process (in steady state), with an auto-correlation function given by: $R_{\theta}(\tau) = \mathbb{E}\{ \theta(t) \theta(t + \tau)\} = \frac{\sigma_{\theta}^2}{2 \eta_{\theta}}e^{- \eta_{\theta} |\tau|}$ \cite{Petrovic2007}. 

\begin{lemma} \label{Lemma_OU_PN_properties}
For phase-noise modeled as an OU process we have:
\begin{subequations} \label{eqn_PN_OU_lemma}
\begin{flalign}
& \qquad \sum_{k \in \mathcal{K}} \Omega[k] {\Omega[k+k_1]}^{*} = \delta^{K}_{0,k_1}, & \label{eqn_OU_lemma_1} \\
& \qquad \Delta_{k_1,k_2} \triangleq \mathbb{E}\{\Omega[k_1] {\Omega[k_2]}^{*}\} & \nonumber \\
& \qquad \qquad \ \approx \frac{\delta^{K}_{k_1,k_2} e^{-R_{\theta}[0]}}{K} \sum_{u = -\lfloor K/2 \rfloor}^{\lfloor (K-1)/2 \rfloor} e^{R_{\theta}[u]} e^{-{\rm j}2 \pi \frac{k_1 u}{K}} & \label{eqn_OU_lemma_3} 
\end{flalign}
for arbitrary integers $k_1,k_2$, where $\delta^{K}_{a,b} = 1$ if $a=b \ ({\rm mod} \ K)$ or $\delta^{K}_{a,b} = 0$ otherwise, and $R_{\theta}[n] \triangleq \mathbb{E}\{ \theta[\dot{n}] \theta[\dot{n}+n]\} = \frac{\sigma_{\theta}^2}{2 \eta_{\theta}}e^{- \eta_{\theta} |n T_{\rm s}/K|}$. 
\end{subequations}
\end{lemma}
\begin{proof}[Proof of Lemma \ref{Lemma_OU_PN_properties}]
Note that from the definition of $\Omega[k]$, we have $e^{- {\rm j}\theta[n]} \stackrel{\mathcal{F}}{\longrightarrow} \Omega[k]$ and $e^{{\rm j}\theta[n]} \stackrel{\mathcal{F}}{\longrightarrow} \Omega^{*}[-k]$, where $\mathcal{F}$ represents the nDFT Operation. Then using convolution property of the nDFT, we have:
\begin{eqnarray}
e^{- {\rm j}\theta[n]} e^{{\rm j}\theta[n]} \stackrel{\mathcal{F}}{\longrightarrow} \sum_{a \in \mathcal{K}} \Omega[a] \Omega^{*}[a+k] \nonumber \\
\Rightarrow 1 \stackrel{\mathcal{F}}{\longrightarrow} \sum_{a \in \mathcal{K}} \Omega[a] \Omega^{*}[a+k] \nonumber \\
\Rightarrow \delta_{0,k}^{K} = \sum_{a \in \mathcal{K}} \Omega[a] \Omega^{*}[a+k] \nonumber
\end{eqnarray}
which proves property \eqref{eqn_lemma_1}. Property \eqref{eqn_lemma_3} can be obtained as follows:
\begin{eqnarray}
\Delta_{k_1,k_2} & \triangleq & \mathbb{E}\{\Omega[k_1] {\Omega[k_2]}^{*}\} \nonumber \\
&=& \frac{1}{K^2} \sum_{\dot{n}, \ddot{n} = 0}^{K-1} \mathbb{E}\{ e^{-{\rm j}[\theta[\dot{n}]-\theta[\ddot{n}]} \} e^{-{\rm j}2 \pi \frac{[k_1 \dot{n} - k_2 \ddot{n}]}{K}}  \nonumber \\
&\stackrel{(2)}{=}& \frac{1}{K^2} \sum_{\ddot{n} = 0}^{K-1} \sum_{u = -\ddot{n}}^{K-1-\ddot{n}} e^{-R_{\theta}[0] + R_{\theta}[u]} e^{-{\rm j}2 \pi \frac{[k_1 u + (k_1- k_2) \ddot{n}]}{K}}  \nonumber \\
& \stackrel{(3)}{\approx} & \frac{1}{K^2} \sum_{\ddot{n} = 0}^{K-1} \sum_{u = -\lfloor K/2 \rfloor}^{\lfloor (K-1)/2 \rfloor} e^{-R_{\theta}[0] + R_{\theta}[u]} e^{-{\rm j}2 \pi \frac{[k_1 u + (k_1- k_2) \ddot{n}]}{K}}  \nonumber \\
& = & \frac{\delta_{k_1,k_2}^{K} e^{-R_{\theta}[0]}}{K} \sum_{u = -\lfloor K/2 \rfloor}^{\lfloor (K-1)/2 \rfloor} e^{R_{\theta}[u]} e^{-{\rm j}2 \pi \frac{k_1 u}{K}} \label{eqn_lemma1_OU_proof1}
\end{eqnarray}
where $\stackrel{(2)}{=}$ follows from similar steps to \eqref{eqn_lemma1_proof1} and $\stackrel{(3)}{\approx}$ follows by noting that $R_{\theta}[u]$ has a limited support around $u=0$ and hence $R_{\theta}[u] \approx R_{\theta}[u-K] \approx 0$ for $u > (K-1)/2$. Note that since $e^{- R_{\theta}[0] + R_{\theta}[u]}$ is an auto-correlation function, its nDFT is non-negative, thus ensuring that $\Delta_{k_1,k_1} \geq 0$ in \eqref{eqn_lemma1_OU_proof1}. 
\end{proof}

\end{appendices}

\bibliographystyle{ieeetr}
\bibliography{IEEEabrv,references_trial}

\begin{IEEEbiography}[{\includegraphics[width=1in,height=1.25in,clip,keepaspectratio]{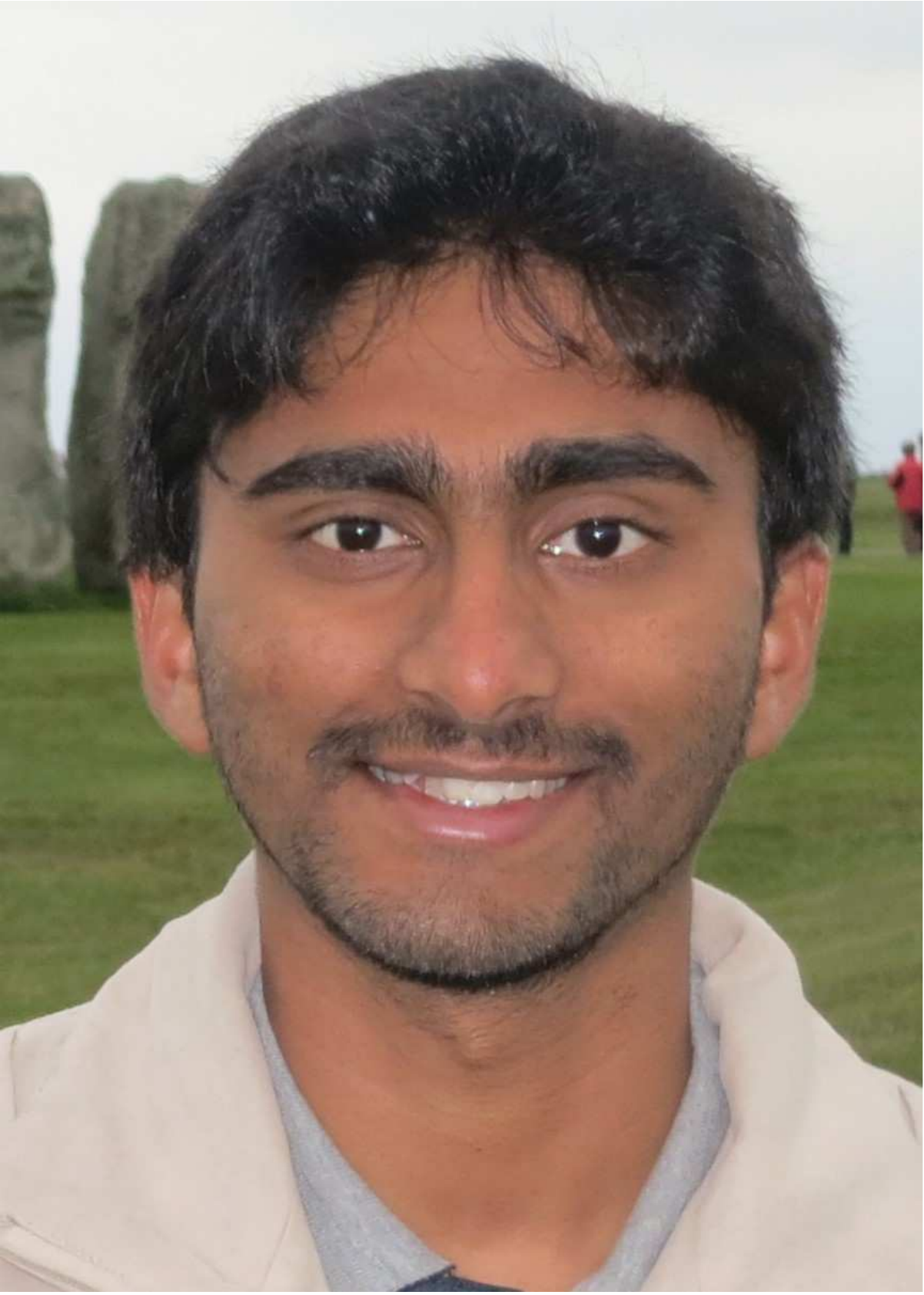}}]{Vishnu V. Ratnam}
(S'10--M'19) received the B.Tech. degree (Hons.) in electronics and electrical communication engineering from IIT Kharagpur, Kharagpur, India in 2012, where he graduated as the Salutatorian for the class of 2012. He received the Ph.D. degree in electrical engineering from University of Southern California, Los Angeles, CA, USA in 2018. He is currently a senior research engineer in the Standards and Mobility Innovation Lab at Samsung Research America, Plano, Texas, USA. His research interests are in AI for wireless, mm-wave and Terahertz communication, massive MIMO, channel estimation and manifold signal processing, and resource allocation problems in multi-antenna networks. 

Dr. Ratnam is a recipient of the Best Student Paper Award at the IEEE International Conference on Ubiquitous Wireless Broadband (ICUWB) in 2016, the Bigyan Sinha memorial award in 2012 and is a member of the Phi-Kappa-Phi honor society.
\end{IEEEbiography}

\begin{IEEEbiography}[{\includegraphics[width=1in,height=1.25in,clip,keepaspectratio]{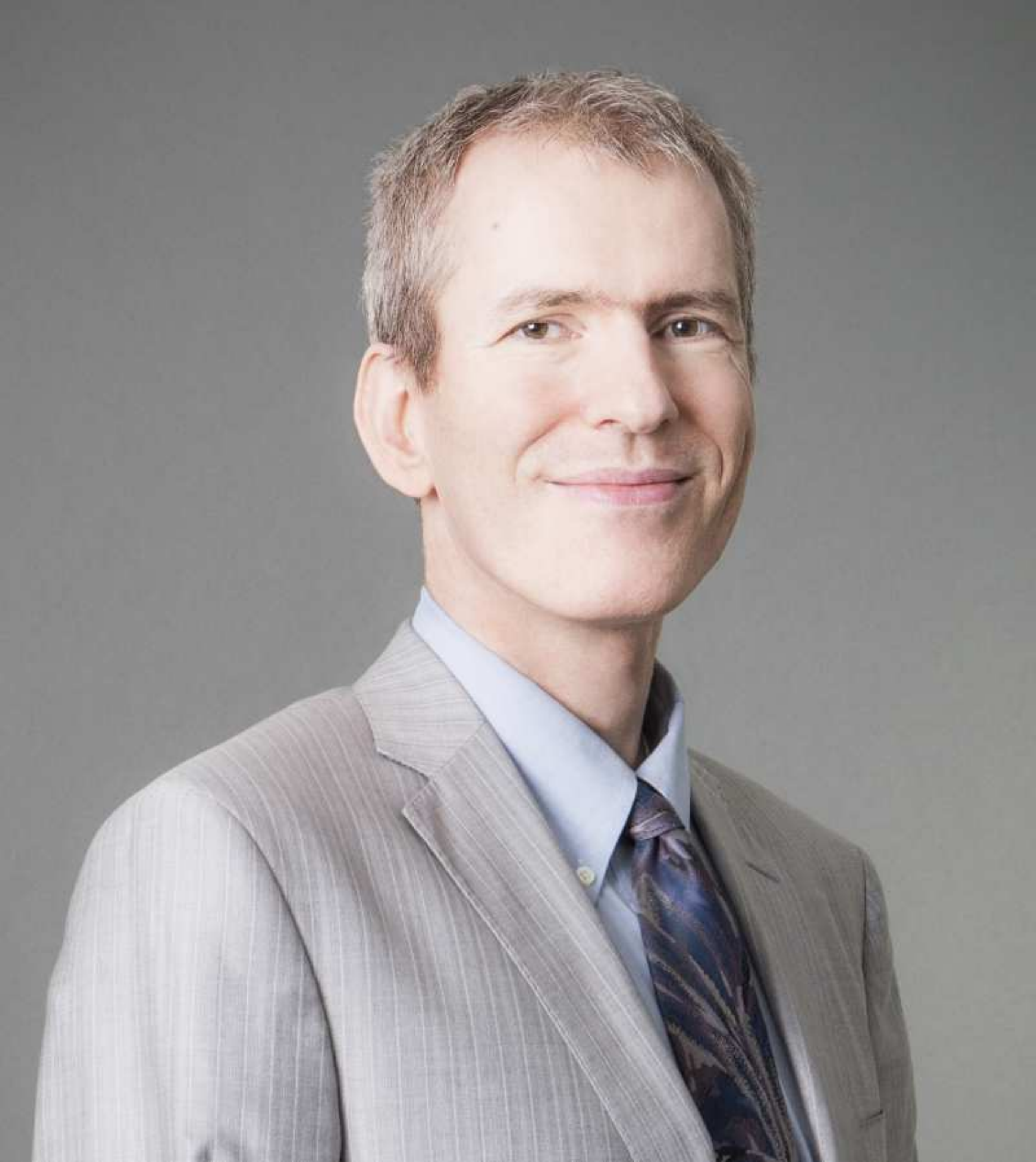}}]{Andreas F. Molisch}
(S'89--M'95--SM'00--F'05) received the Dipl. Ing., Ph.D., and habilitation degrees from the Technical University of Vienna, Vienna, Austria, in 1990, 1994, and 1999, respectively. He subsequently was with FTW (Austria), AT\&T (Bell) Laboratories Research (USA); Lund University (Sweden), and Mitsubishi Electric Research Labs (USA). He is now a Professor and the Solomon Golomb -- Andrew and Erna Viterbi Chair at the University of Southern California, Los Angeles, CA, USA. 

His current research interests are the measurement and modeling of mobile radio channels, multi-antenna systems, wireless video distribution, ultra-wideband communications and localization, and novel modulation formats. He has authored, coauthored, or edited four books (among them the textbook Wireless Communications, Wiley-IEEE Press), 20 book chapters, more than 250  journal papers, more than 340 conference papers, as well as more than 80 patents and 70 standards contributions.

Dr. Molisch has been an Editor of a number of journals and special issues, General Chair, Technical Program Committee Chair, or Symposium Chair of multiple international conferences, as well as Chairman of various international standardization groups. He is a Fellow of the National Academy of Inventors, Fellow of the AAAS, Fellow of the IET, an IEEE Distinguished Lecturer, and a Member of the Austrian Academy of Sciences. He has received numerous awards, among them the Donald Fink Prize of the IEEE, the IET Achievement Medal, the Armstrong Achievement Award of the IEEE Communications Society, and the Eric Sumner Award of the IEEE.
\end{IEEEbiography}

\end{document}